\documentclass[journal]{IEEEtran}

\IEEEoverridecommandlockouts

\usepackage{algorithm}
\usepackage{algorithmic}
\usepackage{amsfonts}
\usepackage{amsthm}
\usepackage{amssymb}
\usepackage{array}
\usepackage{bbm}
\usepackage{bbold}
\usepackage{bm}
\usepackage{booktabs}
\usepackage{caption}
\usepackage{cite}
\usepackage[cmex10]{amsmath}
\usepackage{diagbox}
\usepackage{enumerate}
\usepackage{epstopdf}
\usepackage{footnote}
\usepackage{framed}
\usepackage{graphicx}
\usepackage{indentfirst}
\usepackage{lineno}
\usepackage{multicol}
\usepackage{multirow}
\usepackage{setspace}
\usepackage{subeqnarray}
\usepackage{subfigure}
\usepackage{threeparttable}
\makesavenoteenv{tabular}
\usepackage{url}
\usepackage{xcolor}
\allowdisplaybreaks[4]

\newcommand{\mb}{\mathbf}
\newcommand{\mbb}{\mathbb}
\newcommand{\mc}{\mathcal}

\newcommand{\mr}{\mathrm}
\newcommand{\ti}{\textit}

\newcommand{\ul}{\underline}
\newcommand{\wh}{\widehat}

\newcommand{\tcr}{\textcolor{red}}

\newcommand{\black}{\color{black}}

\newtheorem{theorem}{Theorem}
\newtheorem{lemma}{Lemma}
\newtheorem{proposition}{Proposition}
\newtheorem{corollary}{Corollary}

\newtheorem{example}{Example}
\newtheorem{remark}{Remark}

\newtheorem{claim}{Claim}

\begin{document}

\title{Outer Channel of DNA-Based Data Storage: Capacity and Efficient Coding Schemes}

\author
{
\IEEEauthorblockN{
Xuan~He,
Yi~Ding,
Kui~Cai,
Guanghui~Song,
Bin~Dai,
and~Xiaohu~Tang}
%Xuan~He$^\dag$,
%Yi~Ding$^\dag$$^\ast$,
%Kui~Cai$^\ast$,
%Guanghui~Song$^\ddag$,
%Bin~Dai$^\dag$,
%and~Xiaohu~Tang$^\dag$}\\
%	\IEEEauthorblockA{${}^\dag$Information Coding and Transmission Key Lab of Sichuan Province, Southwest Jiaotong University, China\\}
%	\IEEEauthorblockA{${}^\ast$Science, Mathematics and Technology (SMT) Cluster, Singapore University of Technology and Design, Singapore\\}
%	\IEEEauthorblockA{${}^\ddag$State Key Lab of Integrated Services Networks, Xidian University, China\\}
%	Email: \{xhe, xhutang, daibin\}@swjtu.edu.cn, dy1226898582@163.com, cai$\_$kui@sutd.edu.sg, songguanghui@xidian.edu.cn
%\thanks{This work was supported by National Key Research and Development Program of China under Grant 2021YFF1200200, by National Natural Science Foundation of China (NSFC) under Grant 62271369, by Fundamental Research Funds for the Central Universities under Grant 2682023CG001, and by SUTD KICKSTARTER INITIATIVE (SKI) Grant SKI 2021\_02\_17.}

\thanks{This work was presented in part at \cite{	ding2023anefficient}. DOI: 10.1109/ICCCWorkshops57813.2023.10233840}% <-this % stops a space
\thanks{X. He, Y. Ding, B. Dai, and X. Tang are with the Information Coding and Transmission Key Lab of Sichuan Province, Southwest Jiaotong University, China.
K. Cai is with the Science, Mathematics and Technology (SMT) Cluster, Singapore University of Technology and Design, Singapore.
G. Song is with the State Key Lab of Integrated Services Networks, Xidian University, China.}% <-this % stops a space
%\thanks{Manuscript received April 19, 2005; revised September 17, 2014.}
}

\maketitle

\begin{abstract}
In this paper, we consider the outer channel for DNA-based data storage.
When transmitting over the outer channel, each DNA string is treated as a unit/symbol that would be either correctly received, or erased, or corrupted by uniformly distributed random symbol substitution errors, and all strings are randomly shuffled  with each other.
We first derive the capacity of the outer channel, which implies that the uniformly distributed random symbol substitution errors are only  as harmful as the erasure errors (for infinite-length non-binary random linear codes with near maximum likelihood decoding).
Next, we propose practically efficient coding schemes which encode the bits at the same position of different strings into a codeword.
We compute the soft/hard information of each bit, which allows us to independently decode the bits within a codeword, leading to an independent decoding scheme.
To improve the decoding performance, we measure the reliability of each string based on the independent decoding result, and perform a further step of decoding over the most reliable strings, leading to a joint decoding scheme.
Simulations with low-density parity-check  codes confirm that the joint decoding scheme can reduce the frame error rate by more than 3 orders of magnitude compared to the independent decoding scheme, and it can outperform the state-of-the-art decoding scheme in the literature across a wide range of parameter regions.
\end{abstract}


\IEEEpeerreviewmaketitle

%%%%%%%%%%%%%%%%%%%%%%%%%%%%%%%%%%%%%%%%%%%%%%%%%%%%%%%%%%%%%%%%%%%%%%%%%%%%%%%%%%%%%%%%%%%%%%%%%%%%%%%%%%%%%%%%%%%%%%%%%%%%%%%%%%%%%

\section{Introduction}\label{section: introduction}

Due to the increasing demand for data storage, DNA-based data storage systems have attracted significant attention, since they can achieve extremely high data storage capacity, with very long life duration and low maintenance cost \cite{erlich2017dna, organick2018random, heckel2019acharacterization, dong2020dna}.
A typical architecture for these systems  \cite{erlich2017dna, organick2018random, heckel2019acharacterization, dong2020dna} is depicted in Fig. \ref{fig: DNA model}.
The $k \times w$ binary source data matrix $\mb{U} \in \mathbb{F}_2^{k \times w}$ is first transformed into $n \times l$ binary matrix $\mb{X} \in \mathbb{F}_2^{n \times l}$ by the outer encoder,  where $\mbb{F}_2 \triangleq \{0, 1\}$ denotes the binary field.
Next, each row of $\mb{X}$ is transformed into an inner codeword by the inner encoder.
Finally, each inner codeword is synthesized (written) as a DNA string/strand/oligo and all DNA strings are stored  in an unordered manner.
A DNA string consists of 4 bases: A, T, C, G.
However, since each base naturally corresponds to two bits, it is equivalent to considering bits for simplicity \cite{he2023basis, shomorony2021DNA}.
When recovering source data, a fraction of DNA strings are first sequenced
(read out) in a random sampling fashion.
Next, the sequenced strings are decoded by the inner decoder, resulting in an estimation $\mb{Z}$ of $\mb{X}$, where $\mb{Z}$ is regarded as having $n$ rows, each of which is either empty (represented by `?') or a length-$l$ bit string (retrieved as information from an inner codeword during the inner decoding).
Finally, the outer decoder takes $\mb{Z}$ to recover the source data.

\begin{figure}[t]
\centering
\vspace{2mm}
\includegraphics[scale = 0.5]{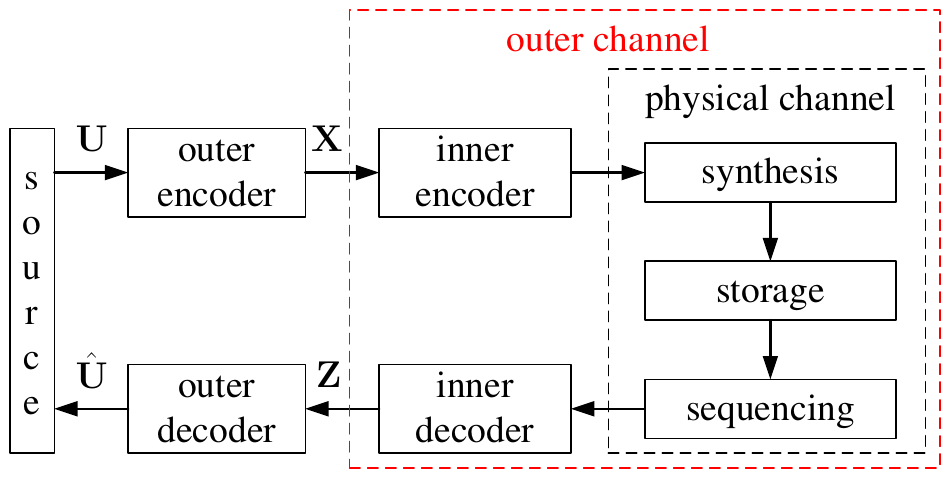}
\caption{
A typical architecture for DNA-based data storage.
}
\label{fig: DNA model}
\end{figure}

In this paper, we refer to the synthesis, storage, and sequencing together as physical channel.
When passing through the physical channel, on the one hand, DNA strings are randomly shuffled with each other, i.e., lose relative order, due to the unordered storage and random sampling.
On the other hand, DNA strings can be missing or corrupted by edit (insertion, deletion, and substitution) errors, and in particular, those with unbalanced GC-content or too long homopolymers are more likely to suffer from such errors \cite{heckel2019acharacterization}.

To resist the edit errors within a DNA string,  one role of the inner code is to work as constrained code, which encodes each row of $\mb{X}$ as a DNA string with desired GC-content and homopolymers, such that the raw error rate of each string can be reduced when passing through the physical channel, i.e., improve the  quality of physical channel.
Another role of the inner code is to work as error-correction code (ECC), which needs to detect/correct the errors within a DNA string.
Note that the inner decoding has a non-negligible possibility of encountering undetectable errors, i.e., generating incorrect inner codewords that would be beyond the error detection capability of any inner code.
The main reasons leading to this problem include that (i) high rate  inner codes with weak error-detection capability are normally employed since  DNA synthesis is very expensive \cite{erlich2017dna, organick2018random, heckel2019acharacterization, dong2020dna} and (ii) the lack of practically efficient inner codes for detecting multiple  edit errors \cite{mitzenmacher2009asurvey, cheraghchi2021anoverview,  sima2020optimal}.
Therefore,  the non-inner codewords generated by the inner decoding can be simply discarded, leading to a part of the empty rows of $\mb{Z}$ (the other part of the empty rows of $\mb{Z}$ corresponds to the missing strings caused by the physical channel); the inner codewords generated by the inner decoding are used to retrieve information to form all the non-empty rows of $\mb{Z}$,  where the rows retrieved from  incorrect inner codewords (undetectable errors) would suffer from burst bit  substitution errors.

According to the above  discussions, the matrix $\mb{Z}$, as the output of the inner decoder and the input of the outer decoder, must satisfy the following two facts:
\begin{enumerate}[{Fact} 1:]
\item   Each row of $\mb{Z}$, relative to the original one of $\mb{X}$ it comes from, has exactly three statuses:
    \begin{itemize}
    \item   correct (with a sufficient large probability)
    \item   empty/erased (corresponding to a missing string from the physical channel or a discarded non-inner codeword),
    \item   incorrect with burst bit  substitution errors (corresponding to the undetectable errors from the inner decoding).
    \end{itemize}
\item  The rows of $\mb{Z}$ are unordered, i.e., it is not known that each row of $\mb{Z}$ comes from which row of $\mb{X}$.
\end{enumerate}

Therefore, to ensure the reliability of DNA-based data storage, the role of the outer code is to work as ECC, which needs to address the above problems of  $\mb{Z}$.
In this paper, we focus on the outer code.
To simplify the discussion on it, following \cite{he2023basis}, we combine the inner code and the physical channel together into an equivalent channel, called outer channel in this paper, as shown in Fig. \ref{fig: DNA model}.
Following \cite{he2023basis} again, we  further model the outer channel as the concatenation of two sub-channels: channel-1 and channel-2, as shown in Fig. \ref{fig: Channel model}.

\begin{figure}[t]
\centering
\vspace{2mm}
\includegraphics[scale = 0.5]{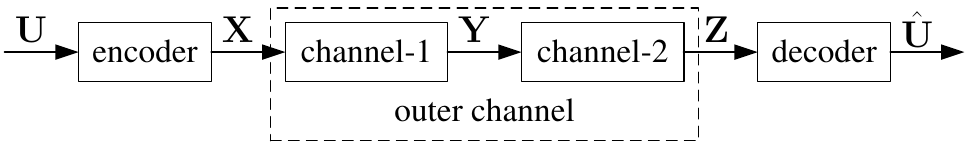}
\caption{
System model for the outer channel, where $\mb{U} \in \mathbb{F}_2^{k \times w}$ is the data matrix, $\mb{X} \in \mathbb{F}_2^{n \times l}$ is the encoded matrix, channel-1 has transition probability given by \eqref{eqn: channel-1}, channel-2 randomly permutates the rows of $\mb{Y}$, and $\wh{\mb{U}}$ is an estimation of $\mb{U}$.
%First, the data matrix $\mb{U}$ is encoded as $\mb{X}$.
%Then, the rows of $\mb{X}$ are transmitted over channel-1 and received one-by-one in order according to \eqref{eqn: channel-1}.
%Next, $\mb{Y}$ is transmitted over channel-2 and  its rows are permuted uniformly at random to form $\mb{Z}$.
%Finally, decoder gives an estimation $\wh{\mb{U}}$ of $\mb{U}$ based on the channel output $\mb{Z}$.
}
\label{fig: Channel model}
\end{figure}

Briefly speaking (see Section \ref{section: system model} for more details), we view the inner code as a part of channel-1 and each row  (NOT each bit)  of $\mb{X}$ is treated as a unit (a symbol of $l$ bits) when passing through channel-1.
Specifically, each row of $\mb{X}$ is correctly received with probability $p_c$, or erased with probability $p_e$, or corrupted by uniformly distributed random symbol substitution errors  (a type of burst bit  substitution errors) with probability $p_s$.
After that, for the output $\mb{Y}$ of channel-1, its rows are randomly shuffled with each other by channel-2 (i.e., channel-2 works as a permutation channel), leading to the final output $\mb{Z}$ of the outer channel.
In this case, the rows of $\mb{Z}$ lose relative order and may be either correct, or erased, or corrupted by burst bit substitution errors, well matching Facts 1 and 2.

In this paper, we are interested in both computing the capacity and developing practically efficient   coding schemes for the outer channel in Fig. \ref{fig: Channel model} (while the outer channel in Fig. \ref{fig: DNA model} is referred to as ``original outer channel" in the rest of this paper).
To be specific, following the classical definition in \cite{cover1999elements}, we define code rate as the average number of information bits stored per coded bit.
Consequently, the code rate in Fig. \ref{fig: Channel model} is
\begin{equation}\label{eqn: code rate}
R \triangleq \frac{kw}{nl}.
\end{equation}
We further state that the rate $R$ is achievable if there exists a family of codes with rate $R$ such  that the decoding error probability tends to $0$ as $n \to \infty$, and the channel capacity refers to the maximum achievable code rate.

\color{black}

\subsection{Related Work}

This is an extended work of its conference version \cite{ding2023anefficient}.
In \cite{ding2023anefficient}, we proposed an efficient coding scheme for the outer channel.
In this work, we further prove the capacity of the outer channel.

The outer channel is almost the same as the concatenation channel in \cite[Fig. 2(a)]{he2023basis}, except that we additionally consider erasure errors in this paper.
In \cite{he2023basis}, fountain codes \cite{byers1998digital, luby2002lt} were used as the outer codes, and an efficient decoding scheme, called basis-finding algorithm (BFA), was proposed for fountain codes.
BFA can also work for the outer channel in this paper when fountain codes are being used, and it serves as a benchmark for the simulation results later in Section \ref{section: simulation}.
However, the capacity of the concatenation channel has not been discussed in \cite{he2023basis}.

In \cite[Fig. 2]{makur2020coding}, Makur considered a noisy permutation channel which essentially is the same as the concatenation channel in \cite[Fig. 2(a)]{he2023basis}.
Define
\begin{equation}
\beta \triangleq \frac{l}{\log_2 n},
\end{equation}
in which $n$ and $l$ are the number of rows and number of columns in $\mb{X}$, respectively.
Makur modified the conventional definition of channel capacity \cite{cover1999elements} and computed the modified capacity for $n \to \infty$ and fixed $l$ (i.e., $\beta \to 0$ in this case).
However, we will soon see that the (conventional) capacity of the outer channel equals zero for $\beta \leq 1$, and we thus aim to compute the  capacity of the outer channel for $\beta > 1$ in this paper.

In  \cite{shomorony2021DNA}, Shomorony and Heckel considered a noise-free shuffling-sampling channel, which can also be described as the concatenation of two sub-channels.
The second sub-channel is the same as our channel-2.
The first sub-channel outputs some noise-free copies of rows of the input $\mb{X}$, where the number of copies follows a certain probability distribution, e.g, the Bernoulli distribution.
Let $p_{\text{erasure}}$ denote the probability that the first sub-channel does not output any copy for a row of $\mb{X}$.
According to \cite[Theorem 1]{shomorony2021DNA}, the noise-free shuffling-sampling channel has capacity
\begin{equation}\label{eqn: Cnf}
    C_{nf} = (1 - p_{\text{erasure}})(1 - 1/\beta),
\end{equation}
as long as $\beta > 1$, and  $C_{nf} = 0$ for $\beta \leq 1$.
In \eqref{eqn: Cnf},  $1 - 1/\beta$ can be understood as the loss due to the random permutation of the second sub-channel, which implies that index-based coding scheme that uses $\log_2 n$ bits for uniquely labelling each row of $\mb{X}$ is optimal \cite{shomorony2021DNA}.

Shomorony and Heckel \cite{shomorony2021DNA} additionally studied a noisy shuffling-sampling channel.
In \cite{weinberger2022dna} and  \cite{lenz2023noisy},  noisy drawing channels were considered.
However,  these channels include  discrete memoryless channels (DMCs, e.g., binary symmetric channel (BSC)) as a part, which add independent and identically distributed (i.i.d.) substitution errors to each bit within a row of $\mb{X}$.
This is the key difference from our channel-1 which does not independently add substitution errors to each bit within a row.
We will later explain that, compared to the DMCs in \cite{shomorony2021DNA, weinberger2022dna, lenz2023noisy}, our channel-1 can better match Fact 1.

Our channel-1 belongs to the class of noise-erasure channels (NECs), for which Song \ti{et al.} \cite{song2018capacity} rigorously derived the channel capacity.
It is crucial to clarify that, the outer channel, as the concatenation of channel-1 and a permutation channel (i.e., channel-2), is fundamentally different from a standalone NEC investigated in \cite{song2018capacity, weidmann2012fresh}.
In fact, for a channel formed by the concatenation of a noisy channel and a permutation channel, the permutation channel  introduces intrinsic complexities in capacity derivation, even when the component noisy channel is well-studied.
For example, the noisy shuffling-sampling channel \cite{shomorony2021DNA} is the concatenation of a BSC and a permutation channel.
Although the BSC is well-studied and its capacity is well-known, the capacity of the noisy shuffling-sampling channel, as presented in \cite[Theorem 2]{shomorony2021DNA}, is only solvable  within a limited parameter range.
Moreover, Weinberger \ti{et al.} \cite{weinberger2022dna} and  Lenz \ti{et al.} \cite{lenz2023noisy}  also were only able to  partially obtain the capacity of noisy drawing channels, where the noisy shuffling-sampling channel is included as a special case.
Given that our outer channel follows the same concatenated architecture, deriving its capacity represents a fortunate but highly non-trivial achievement.

We summarize the differences between the existing work and this work in Table \ref{table: comparison}.
As can be seen from Table \ref{table: comparison},  the capacity of the outer channel is not clear for $\beta > 1$ up till now.

\begin{table}[t]
	\begin{center}
\caption{Comparison between the existing work and this work}
\label{table: comparison}
	\renewcommand\arraystretch{1.3}
		\begin{tabular}{|c|p{2.7cm}|p{2.8cm}|}
			\hline
            Work & Channel model & Difference\\
            \hline
            this work & outer channel & \\\hline
            \cite{ding2023anefficient} & outer channel & no capacity result\\\hline
            \cite{he2023basis} & concatenation channel & no erasure errors; no capacity result\\\hline
            \cite{makur2020coding} & noisy permutation channel & no erasure errors; modified capacity for $n \to \infty$ and fixed $l$ (i.e., $\beta \to 0$)\\\hline
            \cite{shomorony2021DNA} & noise-free shuffling-sampling channel & no substitution errors \\\hline
            \cite{shomorony2021DNA, weinberger2022dna, lenz2023noisy} & noisy  shuffling-sampling / drawing channels & bits within a string suffer from i.i.d. substitution errors\\\hline
            \cite{song2018capacity, weidmann2012fresh} & noise-erasure channels & not concatenating a permutation channel\\
            \hline
		\end{tabular}
	\end{center}
\end{table}

\subsection{Main Contributions}

This paper provides two main contributions.
Our first contribution is to derive the capacity $C$ of the outer channel, given by the following theorem.

\begin{theorem}\label{theorem: capacity}
The capacity of the outer channel is
\begin{equation}\label{eqn: C}
    C = p_c(1 - 1/\beta),
\end{equation}
as long as $\beta > 1$, and $C = 0$ for $\beta \leq 1$.
\end{theorem}

From \cite{ shomorony2021DNA}, it is known that $C = 0$ for $\beta \leq 1$.
Therefore,  we focus on $\beta > 1$ by default from now on except where otherwise stated.
Intuitively, the received rows in $\mb{Z}$ that are corrupted by uniformly distributed random symbol substitution errors can provide no (at most negligible) information about $\mb{X}$.
Suppose there is  a genie-aided outer channel which has a genie to identify and remove these incorrect rows.
In this case, the channel becomes a type of noise-free  shuffling-sampling channel in \cite{shomorony2021DNA} with $p_{\text{erasure}} = p_e + p_s$.
According to \eqref{eqn: Cnf}, its capacity is given by $C_{nf} = (1 - p_{\text{erasure}})(1 - 1/\beta) = p_c (1 - 1/\beta) $.
In contrast to \eqref{eqn: C}, it implies that the uniformly distributed random symbol substitution errors are only  as harmful as erasure errors.
Note that this conclusion follows from infinite-length non-binary random linear codes with near maximum likelihood (ML) decoding (used in the proof of Theorem \ref{theorem: capacity}).
However, for finite-length binary linear codes with practically efficient decoding, it is well known, also shown by our simulations in Section \ref{section: simulation}, that the former can be more harmful than the latter.
We will discuss more on this issue later in Section \ref{section: simulation}.
Theorem \ref{theorem: capacity} motivates us to develop efficient decoding schemes in Section \ref{section: Decoding Scheme} by first converting incorrect rows to erased ones, i.e., first identify and remove the incorrect rows.

Our second contribution is to develop an efficient coding scheme for the outer channel.
More specifically, our encoding scheme is first to choose a block ECC to encode each column of $\mb{X}$ as a codeword so as to tackle erasure and substitution errors from channel-1, and then to add a unique address to each row of $\mb{X}$ so as to combat disordering of channel-2.
Our decoding scheme is first to derive the soft/hard information for each  data bit in $\mb{X}$.
It allows us to decode each column of $\mb{X}$ independently, leading to an independent decoding scheme.
However, this scheme may not be very efficient since it requires the successful decoding of all columns to fully recover $\mb{X}$.
%However, this scheme ignores the correlation between bits within a string.
Therefore, we further measure the reliability of the received rows of $\mb{X}$ based on the independent decoding result.
Then, we take the most reliable received rows to recover $\mb{X}$ like under the erasure channel, leading to an efficient joint decoding scheme.
Simulations with low-density parity-check (LDPC) codes \cite{Gallager62} show that the joint decoding scheme can reduce the frame error rate (FER) by more than 3 orders of magnitude compared to the independent decoding scheme.
Moreover, it demonstrates that the joint decoding scheme and the BFA can outperform each other under different parameter regions.

\subsection{Organization}

The remainder of this paper is organized as follows.
Section \ref{section: system model} illustrates the model for the outer channel, and Section \ref{section: encoder and ML decoder} further gives an encoder as well as an ML decoder for the outer channel.
Section \ref{section: capacity} proves Theorem \ref{theorem: capacity}.
Section \ref{section: Decoding Scheme} develops both the independent and joint decoding schemes.
Section \ref{section: simulation} presents the simulation results.
Finally, Section \ref{section: conclusion} concludes this paper.

\subsection{Notations}

In this paper, we generally use non-bold lowercase letters for scalars (e.g., $n$), bold lowercase letters for (row) vectors (e.g., $\mb{x}$), bold uppercase letters for matrices (e.g., $\mb{X}$), non-bold uppercase letters for random variables and events (e.g., $E$), and calligraphic letters for sets\footnote{This paper only considers standard sets, i.e., sets without repeated elements. That is, if repeated elements are inserted into a set, then only one  of them is kept.} (e.g., $\mathcal{Z}$).
For any positive integer $q$ which is a power of a prime, denote $\mbb{F}_{q}$ as the Galois field $\mr{GF}(q)$.
For any non-negative integer $n$, denote  $[n] \triangleq  \{1, 2, \ldots, n\}$.
For any $n \times l$ matrix $\mb{X}$, we refer to its $i$-th row and $(i,j)$-th entry by $\mb{x}_i$ and $x_{i, j}$, respectively, i.e., $\mb{X} = [(\mb{x}_i^{\mr{T}})_{1 \leq i \leq n}]^{\mr{T}} = ({x}_{i,j})_{1 \leq i \leq n, 1 \leq j \leq l}$, where $(\cdot)^{\mr{T}}$ is the transpose of a vector or matrix.
For any two matrices $\mb{X}$ and $\mb{Y}$, supposing $\mc{R}(\mb{X})$ and $\mc{R}(\mb{Y})$ are the sets consisting of the row vectors of $\mb{X}$ and $\mb{Y}$, respectively,  define $\mb{X} \cap \mb{Y} \triangleq \mc{R}(\mb{X}) \cap \mc{R}(\mb{Y}) $, e.g.,
\[
\mb{X} \cap \mb{Y} = \{(0,0), (1, 1)\} \text{~for~} \mb{X} = \begin{bmatrix}
  0 & 0\\
  0 & 0\\
  1 & 1\\
  0 & 1\\
\end{bmatrix},
\mb{Y} = \begin{bmatrix}
  0 & 0\\
  0 & 0\\
  1 & 1\\
  1 & 1\\
\end{bmatrix}.
\]
For any two events $E_1$ and $E_2$, denote $E_1 \vee E_2$ as  their union, and denote $\mbb{P}(E_1)$ as the probability that $E_1$ happens.
For any set $\mathcal{Z}$, denote its cardinality by $|\mc{Z}|$.
The question mark `?' can denote an unknown bit or row vector when the context is clear.

\section{System Model}\label{section: system model}

In this section, we give detailed explanations to the system model in Fig. \ref{fig: Channel model}.
Specifically, the  data matrix $\mb{U} \in \mbb{F}_2^{k \times w}$ is first encoded as $\mb{X}\in \mathbb{F}_2^{n \times l}$.
Next, $\mb{X}$ is transmitted, with each row being a unit (a symbol of $l$ bits), over the outer channel and $\mb{Z}$ is the output.
Finally, the decoding is performed over $\mb{Z}$ to give an estimation $\wh{\mb{U}}$ of $\mb{U}$.
The encoder and decoder in Fig. \ref{fig: Channel model} respectively correspond to the outer encoder and decoder in Fig. \ref{fig: DNA model}. The outer channel  in Fig. \ref{fig: Channel model}, which is modelled as the concatenation of channel-1 and channel-2, represents a simplification of the original outer channel in Fig. \ref{fig: DNA model}.
%In particular, channel-1 characterizes the erasure errors and correlated substitution errors while channel-2 characterizes the unordered manner of the original outer channel.

We first explain channel-1 in detail.
The input and output of channel-1 are $\mb{X}$ and $\mb{Y}$, respectively, where for each $i \in [n]$, $\mb{y}_i$  is the channel output for transmitting $\mb{x}_i$ over channel-1.
We model the  transition probability of channel-1 by
\begin{equation}
    \label{eqn: channel-1}
\mbb{P}(\mb{y}_i | \mb{x}_i) =
    \begin{cases}
        p_c, & \mb{y}_i = \mb{x}_i, \\
        p_e, & \mb{y}_i = ?, \\
        p_s / (2^l - 1), & \mb{y}_i = \mb{e} \in \mathbb{F}_2^l \setminus \{\mb{x}_i\},
    \end{cases}
\end{equation}
where `?' represents an erasure error (unknown/empty row vector), $p_c + p_e + p_s = 1$, and
\begin{equation}\label{eqn: pc > ps}
    p_c > p_s / (2^l - 1).
\end{equation}

To understand \eqref{eqn: channel-1}, recall that when transmitting over the original outer channel in Fig. \ref{fig: DNA model}, $\mb{x}_i$ will sequentially pass through the inner encoder, physical channel, and inner decoder such that the corresponding outer channel output is an estimation of $\mb{x}_i$ outputted by the inner decoder.
Eqn. \eqref{eqn: channel-1} exactly treats $\mb{y}_i$ as the estimation of $\mb{x}_i$ outputted by the inner decoder, which thus can be classified into the following three cases (as what Fact 1 stated):
\begin{itemize}
\item   $\mb{y}_i = \mb{x}_i$: The inner decoder successfully outputs the correct estimation of $\mb{x}_i$. It happens with probability $p_c$ in \eqref{eqn: channel-1}.
\item   $\mb{y}_i = ?$: The inner decoder fails to output any valid estimation of $\mb{x}_i$  such that $\mb{y}_i$ is regarded as an erasure error.
    It happens with probability $p_e$ in \eqref{eqn: channel-1}.
    This case can be caused by two reasons: (i)  the inner codeword corresponding to $\mb{x}_i$ is totally missing when passing through the physical channel such that the inner decoder is impossible to generate a valid estimation of $\mb{x}_i$, and (ii)  the inner codeword corresponding to $\mb{x}_i$ suffers too many edit errors from the physical channel such that the inner decoder cannot generate any codeword and then a valid estimation of $\mb{x}_i$.
\item   $\mb{y}_i \in \mathbb{F}_2^l \setminus \{\mb{x}_i\}$: The inner decoder outputs a valid but incorrect estimation $\mb{y}_i$ of $\mb{x}_i$, i.e., the inner decoding encounters an  undetectable error.
    As in this case $\mb{y}_i$ is retrieved as information string from an inner codeword with undetectable edit errors, it can be any vector in $\mathbb{F}_2^l \setminus \{\mb{x}_i\}$.
    More precisely, $\mb{y}_i$ generally would heavily differ from $\mb{x}_i$ in terms of Hamming distance (bit-level by default) due to two reasons:
    (i) $\mb{x}_i$ may be corrupted by an equal number of insertion and deletion errors which is equivalent to a large number of substitution errors.
    For example, the string `01010101' can be transformed into `10101010' by deleting the leftmost `0' and inserting an `1' at the rightmost position, but their Hamming distance is eight.
    (ii) Many constrained codes (one role of the inner code) related to GC-content and homopolymers, such as the capacity-achieving ones in \cite{	liu2022capacity} and \cite{	nguyen2021capacity}, can suffer from severe error propagation, i.e., a single bit of error in the codeword can result in a large number of  errors in the retrieved information.
    As a result, it is reasonable to model the transition probability of channel-1 in this case as the third case of \eqref{eqn: channel-1} where $\mb{y}_i$ takes any vector from  $\mathbb{F}_2^l \setminus \{\mb{x}_i\}$ with an equal probability of $p_s / (2^l - 1)$, i.e., $\mb{y}_i$ suffers from uniformly distributed random symbol substitution errors (a type of burst bit  substitution errors).
\end{itemize}
Moreover, \eqref{eqn: pc > ps} is required since the inner decoding should succeed with a large enough probability to ensure a good chance on fully recovering the source data.
Note that for fixed $p_c$ and $p_s$, \eqref{eqn: pc > ps} always holds as $l \to \infty$.

We continue to explain channel-2.
The input and output of channel-2 are $\mb{Y}$ and $\mb{Z}$, respectively.
For convenience,  similar to $\mb{Z}$, we regard $\mb{Y}$ as  having $n$ rows in which any erased rows are represented by `?'.
The channel-2 shuffles the rows of $\mb{Y}$ uniformly at random, i.e.,  $\mb{Z}$ takes any permutation of the rows of $\mb{Y}$ with an equal probability of $1/n!$.
Therefore, channel-2 well matches Fact 2.

According to the above discussions, channel-1 and channel-2 can solely match Facts 1 and 2, respectively.
Therefore, their concatenation can well simulate the function of  the outer channel and further greatly simplify the characterization of the outer channel.
This is the main purpose for using concatenation channel.

\begin{remark}
In \cite{shomorony2021DNA, weinberger2022dna, lenz2023noisy}, DMCs are used to characterize the substitution errors for an entire DNA-based storage system.
However, if only modelling the outer channel, it is not sufficiently accurate to model the substitution errors via a DMC which adds i.i.d. substitution errors to each bit of the rows of $\mb{Z}$, since Fact 1 (each incorrect row of $\mb{Z}$ should suffer from burst bit substitution errors) is violated in this case.
\end{remark}

\section{Encoder and Maximum Likelihood Decoder}\label{section: encoder and ML decoder}

In this section, we develop an efficient encoder and an ML decoder for the outer channel.

\subsection{Encoder}

To protect against the errors introduced by the outer channel, we adopt the encoder illustrated by Fig. \ref{fig: encoding} to convert the data matrix $\mb{U} \in \mathbb{F}_2^{k \times w}$ into the encoded matrix $\mb{X} \in \mathbb{F}_2^{n \times l}$ with the following two steps:
\begin{enumerate}[Step 1: ]
\item   An ECC  is chosen to encode $\mb{U} \in \mathbb{F}_2^{k \times w}$ as $\mb{V} \in \mathbb{F}_2^{n \times w}$, so as to tackle the erasure errors and substitution errors from channel-1.
    Random $(n, k)$ linear ECCs over $\mbb{F}_{2^w}$ will be used for proving  Theorem \ref{theorem: capacity}, while any  excellent binary $(n, k)$ linear ECCs (e.g., LDPC codes) can be used to independently  encode each column of $\mb{U}$ as a column of $\mb{V}$ to achieve a good trade-off between efficiency and simplicity in practical scenarios.
\item   A unique address of bit-width $a \triangleq l - w \geq \lceil \log_2 n \rceil$ is added to the tail of each row of $\mb{V}$, leading to $\mb{X}$, so as to combat the disordering of channel-2.
    Without loss of generality (WLOG), we set the address of $\mb{x}_i$ as $i$ for any $i \in [n]$.
    Note that given $n$, increasing $a$ can reduce the probability that an address is changed to a valid address when substitution errors occur.
    Since for DNA-based data storage, we  can simply discard the received rows without valid addresses, increasing $a$ is equivalent to reducing $p_s$ and increasing $p_e$ in our system model.
    Thus,  we use the minimum $a$ by default, i.e., $a = \lceil \log_2 n \rceil$.
    By convenience, we refer to the first $w$ columns and the last $a$ columns of $\mb{X}$ by data  and address, respectively.
\end{enumerate}
According to the above discussions, the second encoding step is fixed such that the overall encoding scheme would be clear once the ECC used in the first encoding step is given.

\begin{figure}[t]
\centering
\includegraphics[scale = 0.5]{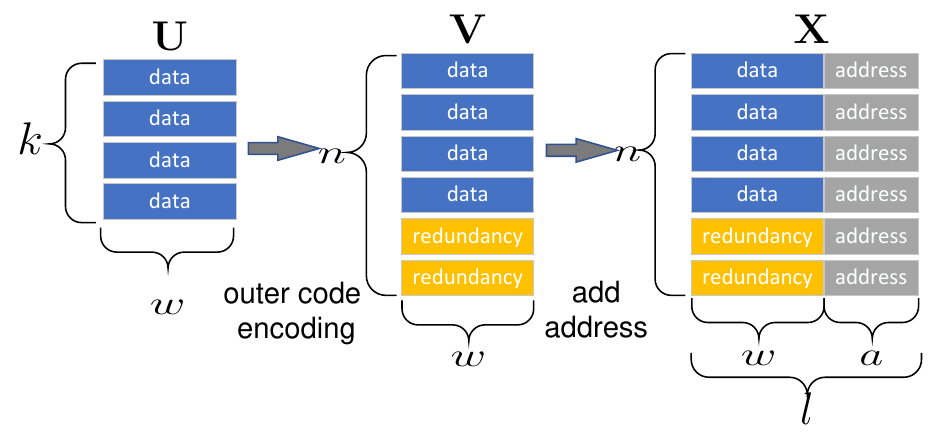}
\caption{
Encoder for the outer channel: first, an ECC is chosen to encode $\mb{U}\in \mathbb{F}_2^{k \times w}$ as $\mb{V}\in \mathbb{F}_2^{n \times w}$;
next,  for each $i \in [n]$, the  address $i$ of bit-width $a = l - w = \lceil \log_2 n \rceil$ is added to the tail of $\mb{v}_i$ to form $\mb{x}_i$, where $\mb{v}_i$ and $\mb{x}_i$ are the $i$-th row of $\mb{V}$ and $\mb{X}$, respectively.}
\label{fig: encoding}
\end{figure}

\color{black}

\begin{example}\label{example: encoding}
This example illustrates the encoding scheme in Fig. \ref{fig: encoding} with a binary linear ECC.
Consider the following source data matrix:
\begin{equation}
    \mb{U} =
    \left[
        \begin{array} {c c c c}
            0 & 0 & 1 & 1\\
            0 & 1 & 0 & 1\\
         \end{array}
    \right].
\end{equation}
%We encode each column of $\mb{U}$ by a $(7, 4)$ hamming code \cite{ECC04} with the following generator matrix:
%\begin{equation}\label{eqn: G}
%    \mb{G} =
%    \left[
%        \begin{array} {c c c c c c c}
%            1 & 0 & 0 & 0 & 1 & 1 & 1\\
%            0 & 1 & 0 & 0 & 1 & 1 & 0\\
%            0 & 0 & 1 & 0 & 1 & 0 & 1\\
%            0 & 0 & 0 & 1 & 0 & 1 & 1\\
%         \end{array}
%    \right].
%\end{equation}
We systematically encode each column of $\mb{U}$ by an $(n = 6, k =  2)$ linear ECC with a minimum distance of $4$ and with the following parity-check matrix:
\begin{equation}\label{eqn: H}
    \mb{H} =
    \left[
        \begin{array} {c c c c c c c}
            1 & 0 & 1 & 0 & 0 & 0\\
            1 & 1 & 0 & 1 & 0 & 0\\
            1 & 1 & 0 & 0 & 1 & 0\\
            0 & 1 & 0 & 0 & 0 & 1\\
         \end{array}
    \right].
\end{equation}
We then get
\begin{equation}\label{eqn: V eg encoding}
    \mb{V} =
    \left[
        \begin{array} {c}
            \mb{U}\\\hline
            \mb{P}\\
        \end{array}
    \right]
    =
%    \mb{G}^{\mr{T}} \mb{U} =
    \left[
        \begin{array} {c c c c}
            0 & 0 & 1 & 1\\
            0 & 1 & 0 & 1\\\hline
            0 & 0 & 1 & 1\\
            0 & 1 & 1 & 0\\
            0 & 1 & 1 & 0\\
            0 & 1 & 0 & 1\\
        \end{array}
    \right],
\end{equation}
where $\mb{P}$ denotes the parity-check part.
For each $i \in [n]$, adding address $i$ (of $a = 3$ bits) into the $i$-th row of $\mb{V}$ leads to
\begin{equation}\label{eqn: X eg encoding}
    \mb{X} = [\mb{V|\mb{A}}] =
    \left[
        \begin{array} {c c c c | c c c}
            0 & 0 & 1 & 1 & 0 & 0 & 1\\
            0 & 1 & 0 & 1 & 0 & 1 & 0\\\hline
            0 & 0 & 1 & 1 & 0 & 1 & 1\\
            0 & 1 & 1 & 0 & 1 & 0 & 0\\
            0 & 1 & 1 & 0 & 1 & 0 & 1\\
            0 & 1 & 0 & 1 & 1 & 1 & 0\\
        \end{array}
    \right],
\end{equation}
where $\mb{A}$ denotes the address part.
\end{example}

\subsection{Maximum Likelihood Decoder}

Following the classical definition in \cite{cover1999elements}, the ML decoder is defined as a decoder that outputs $\wh{\mb{U}}_{\text{ML}} \triangleq \arg\max_{\mb{U} \in \mbb{F}_2^{k \times w}}{\mathbb{P}(\mb{Z}\mid\mb{U})}$ given $\mb{Z}$ as the outer channel output.
Since source data matrix $\mb{U}$ corresponds to encoded matrix $\mb{X}$ in a one-to-one manner, it is equivalent to defining the ML decoder as $\wh{\mb{X}}_{\text{ML}} \triangleq \arg\max_{\mb{X} \in \mc{X}}{\mathbb{P}(\mb{Z}\mid\mb{X})}$, where  $\mc{X} = \{\mb{X}: \mb{X} \text{~can be generated from some ~} \mb{U} \in \mbb{F}_2^{k \times w}\}$.
Recall that the outer channel treats  each row of $\mb{X}$ and $\mb{Z}$ as a unit/symbol.
%Thus for convenience, we regard any $\mb{X}$  (resp. $\mb{Y}, \mb{Z}$) as a length-$n$ non-binary string with symbols from $\mbb{F}_{2^l}$  (resp. $\mbb{F}_{2^l} \cup \{?\}$).
It is straightforward to know that $\wh{\mb{X}}_{\text{ML}}$ must be the encoded matrix that has the least symbol-level Hamming distance to $\mb{Z}$ if channel-2 does not randomly shuffle the rows of $\mb{Y}$, i.e.,  $\wh{\mb{X}}_{\text{ML}}$  has the least symbol-level Hamming distance to $\mb{Y}$.
Therefore, with the random shuffling of $\mb{Y}$, $\wh{\mb{X}}_{\text{ML}}$ is likely to be the encoded matrix that has the largest number of same symbols to $\mb{Z}$ (viewing each matrix as a set consisting of its symbols), i.e., $\wh{\mb{X}}_{\text{ML}} = \arg\max_{\mb{X} \in \mc{X}}{|\mb{X}\cap\mb{Z}|}$.
The following proposition formally states the ML decoder.

\begin{proposition}\label{prop: ML}
Suppose that the encoder in Fig. 3 is adopted and the encoded matrix $\mb{X}$ is sampled from $\mc{X}$ with equal probability.
Then,  if any two non-empty rows of the outer channel output $\mb{Z}$  are distinct, the ML decoder satisfies
\begin{equation}\label{eqn: ML decoder}
\wh{\mb{X}}_{\text{ML}} =  \arg\max_{\mb{X} \in \mc{X}}{|\mb{X}\cap\mb{Z}|}.
\end{equation}
\end{proposition}

\begin{IEEEproof}
For any permutation $\sigma: [n] \to [n]$ and any matrix $\mb{A} = [\mb{a}_i^{\mr{T}}]_{1 \leq i \leq n}^{\mr{T}}$, with a little abuse of notation, denote $\sigma(\mb{A}) \triangleq [\mb{a}_{\sigma(i)}^{\mr{T}}]_{1 \leq i \leq n}^{\mr{T}}$, i.e., for $i \in [n]$, the $i$-th row of $\sigma(\mb{A})$ is the $\sigma(i)$-th row of $\mb{A}$.
Moreover, for any matrix $\mb{B} = [\mb{b}_i^{\mr{T}}]_{1 \leq i \leq n}^{\mr{T}}$, denote $N_c(\mb{A}, \mb{B}) \triangleq |\{i \in [n]: \mb{b}_{i} = \mb{a}_i\}|$, $N_e(\mb{A}, \mb{B}) \triangleq |\{i \in [n]: \mb{b}_{i} = ?\}|$, and $N_s(\mb{A}, \mb{B}) \triangleq |\{i \in [n]: \mb{b}_{i} \notin \{ \mb{a}_i, ?\}\}|$.
Denote $\pi$ as the permutation introduced by channel-2, i.e., for $i \in [n]$, the $i$-th row $\mb{z}_i$ in $\mb{Z}$ corresponds to the transmission of the $\pi(i)$-th row $\mb{x}_{\pi(i)}$ in $\mb{X}$.
Accordingly,  $\pi$ takes any permutation with an equal probability of $\mathbb{P}(\pi) = 1/n!$.
Since $N_c(\pi(\mb{X}), \mb{Z}) + N_e(\pi(\mb{X}), \mb{Z}) + N_s(\pi(\mb{X}), \mb{Z}) = n$ and $N_e(\pi(\mb{X}), \mb{Z})$ is fixed for a given $\mb{Z}$, we have
\begin{align}\label{eqn: Xml = sumpi}
 &\wh{\mb{X}}_{\text{ML}} = \arg\max_{\mb{X} \in \mc{X}}{\mathbb{P}(\mb{Z}\mid\mb{X})} \nonumber
\\= &\arg\max_{\mb{X} \in \mc{X}}{\sum_{\pi}{\mathbb{P}(\mb{Z}\mid\mb{X}, \mb{\pi})\mathbb{P}(\pi)}} \nonumber
\\= &\arg\max_{\mb{X} \in \mc{X}}{\sum_{\pi}{\mathbb{P}(\mb{Z}\mid\mb{X}, \mb{\pi})}} \nonumber
\\= &\arg\max_{\mb{X} \in \mc{X}}{\sum_{\pi}{\prod_{i}} \mbb{P}{(\mb{z}_i\mid \mb{x}_{\pi(i)})}} \nonumber
\\= &\arg\max_{\mb{X} \in \mc{X}}{\sum_{\pi}{p_c^{N_c(\pi(\mb{X}), \mb{Z})}} p_e^{N_e(\pi(\mb{X}), \mb{Z})} \left(\frac{p_s}{2^l - 1}\right)^{N_s(\pi(\mb{X}), \mb{Z})} } \nonumber
\\= &\arg\max_{\mb{X} \in \mc{X}} \sum_{\pi} \left(\frac{p_c(2^l - 1)}{p_s}\right)^{N_c(\pi(\mb{X}), \mb{Z})}.
\end{align}

Given \eqref{eqn: Xml = sumpi}, for showing \eqref{eqn: ML decoder}, it suffices to prove for any $\mb{X}_1,  \mb{X}_2 \in \mc{X}$,
\begin{equation*}\label{eqn: sumpi theta}
\sum_{\pi} \theta^{N_c(\pi(\mb{X}_1), \mb{Z})} > \sum_{\pi} \theta^{N_c(\pi(\mb{X}_2), \mb{Z})}
\Leftrightarrow |\mb{X}_1 \cap \mb{Z}| > |\mb{X}_2 \cap \mb{Z}|,
\end{equation*}
where $\theta \triangleq {p_c(2^l - 1)}/{p_s}$ and $\theta > 1$ by \eqref{eqn: pc > ps}.
Set $r_j=|\mb{X}_j\cap \mb{Z}|$, $j=1,2$.
Denote $\sigma_j$ and $\tau_j$ as rows permutations such that the first $r_j$ rows of $\sigma_j(\mb{X}_j)$ and $\tau_j(\mb{Z})$ are identical.
Note that $\sum_{\pi} \theta^{N_c(\pi(\mb{X}_j), \mb{Z})} = \sum_{\pi} \theta^{N_c(\tau_j \circ \pi \circ \sigma_j^{-1} \circ \sigma_j(\mb{X}_j), \tau_j(\mb{Z}))} = \sum_{\pi} \theta^{N_c( \pi \circ \sigma_j(\mb{X}_j), \tau_j(\mb{Z}))}$, where $\circ$ is the composition of two permutations,  $\sigma_j^{-1}$ is the inverse permutation of $\sigma_j$, and the last equality is because both $\pi$ and $\tau_j \circ \pi \circ \sigma_j^{-1}$ range over all the row permutations.
Thus, it is equivalent to proving
\begin{equation}\label{eqn: sumpi r}
\sum_{\pi} \theta^{N_c( \pi \circ \sigma_1(\mb{X}_1), \tau_1(\mb{Z}))} > \sum_{\pi} \theta^{N_c( \pi \circ \sigma_2(\mb{X}_2), \tau_2(\mb{Z}))}
\Leftrightarrow r_1 > r_2.
\end{equation}

Clearly, for any $\pi$,
\begin{equation}\label{eqn: Nc ri}
N_c( \pi \circ \sigma_j(\mb{X}_j), \tau_j(\mb{Z})) = |\{i \in [r_j]: i = \pi(i)\}|,
\end{equation}
since $|\sigma_j(\mb{X}_j) \cap \tau_j(\mb{Z})| = r_j$ and any two non-empty rows of  $\sigma_j(\mb{X}_j)$ or $\tau_j(\mb{Z})$  are distinct.
By \eqref{eqn: Nc ri}, on the one hand, $\sum_{\pi} \theta^{N_c( \pi \circ \sigma_1(\mb{X}_1), \tau_1(\mb{Z}))} > \sum_{\pi} \theta^{N_c( \pi \circ \sigma_2(\mb{X}_2), \tau_2(\mb{Z}))}
\Rightarrow
\exists \pi, N_c( \pi \circ \sigma_1(\mb{X}_1), \tau_1(\mb{Z})) > N_c( \pi \circ \sigma_2(\mb{X}_2), \tau_2(\mb{Z}))
\Rightarrow
r_1 > r_2$.
On the other hand, $r_1 > r_2 \Rightarrow
\forall \pi, N_c( \pi \circ \sigma_1(\mb{X}_1), \tau_1(\mb{Z})) > N_c( \pi \circ \sigma_2(\mb{X}_2), \tau_2(\mb{Z}))$ and $N_c( \pi' \circ \sigma_1(\mb{X}_1), \tau_1(\mb{Z})) > N_c( \pi' \circ \sigma_2(\mb{X}_2), \tau_2(\mb{Z}))$
where $\pi'$ is the identity permutation
$\Rightarrow
\sum_{\pi} \theta^{N_c( \pi \circ \sigma_1(\mb{X}_1), \tau_1(\mb{Z}))} > \sum_{\pi} \theta^{N_c( \pi \circ \sigma_2(\mb{X}_2), \tau_2(\mb{Z}))}$.
As a result, \eqref{eqn: sumpi r} holds, which completes the proof of Proposition \ref{prop: ML}.
\end{IEEEproof}

Proposition \ref{prop: ML} requires that any two non-empty rows of $\mb{Z}$ are distinct, otherwise \eqref{eqn: Nc ri} may not necessarily hold.
For any distinct $i, j \in [n]$, the probability $p_{ij}$ of that transmitting $\mb{x}_i$ and $\mb{x}_j$ leads to two identical non-empty rows in $\mb{Z}$ is $p_{ij} = p_c \frac{p_s}{ 2^l - 1} + \frac{p_s}{ 2^l - 1} p_c + (2^l - 2)\frac{p_s}{ 2^l - 1} \frac{p_s}{ 2^l - 1}$.
Then, by the union bound, the probability $p$ of that $\mb{Z}$ has at least two identical non-empty rows is $p \leq \sum_{1 \leq i < j \leq n} p_{ij} = \frac{n(n-1)}{2 (2^l - 1)} \left(2p_c p_s + \frac{(2^l - 2) p_s^2}{ 2^l - 1}\right) < \frac{n^2}{2^l} \left(2p_c p_s + p_s^2\right) = n^{2-\beta} \left(2p_c p_s + p_s^2\right)$.
Thus, if $\beta = l/\log_2 n > 2$, $1 - p$ generally approaches $1$,  implying that the event where any two non-empty rows of $\mb{Z}$ are distinct occurs almost surely.
It should be noted that even when this event does not hold in certain cases such that the right-hand side (RHS) of \eqref{eqn: ML decoder} might not strictly qualify as an ML decoder, the RHS is proven to be asymptotically optimal.
This is evidenced by its application (with minor adaptations) in Section \ref{section: capacity} to establish the achievability for proving the outer channel capacity.
However, the RHS of \eqref{eqn: ML decoder} needs to enumerate all $\mb{X} \in \mc{X}$, leading to a prohibitively high complexity of $O(2^{kw})$.
Therefore, it is only of theoretical interest, while we need polynomial-time efficient decoders for practical scenarios, such as those developed in Section \ref{section: Decoding Scheme}.

\section{Proof of Theorem \ref{theorem: capacity}}\label{section: capacity}

Recall that $\beta = \frac{l}{\log_2 n}$ and the code rate of the system in Fig. \ref{fig: Channel model} is $R = \frac{kw}{nl}$.
According to \cite{ shomorony2021DNA}, the outer channel has capacity $C = 0$ for $\beta \leq 1$.
Therefore, to prove Theorem \ref{theorem: capacity}, we only need to further show the following two lemmas.

\begin{lemma}[achievability]\label{lemma: achievable}
For $\beta > 1$ and arbitrarily small fixed $\epsilon > 0$, the following code rate is achievable:
\begin{equation}
    R = (p_c - 2 \epsilon) (1 - 1/\beta).
\end{equation}
That is, there is a family of codes with rate $R = (p_c - 2 \epsilon) (1 - 1/\beta)$ (setting $k = n(p_c - 2\epsilon)$ and $w = l(1 - 1/\beta)$) such that the decoding error probability $\mbb{P}(\wh{\mb{U}} \neq \mb{U}) \to 0$ as $n \to \infty$.
\end{lemma}

\begin{lemma}[converse]\label{lemma: converse}
For $\beta > 1$ and arbitrarily small fixed $\delta > 0$, as $n \to \infty$, any achievable code rate $R$ satisfies
\begin{equation}
    R \leq (p_c + \delta) (1 - 1/\beta).
\end{equation}
\end{lemma}

\subsection{Proof of Lemma \ref{lemma: achievable}}

We consider using a family of codes $\mc{C}_n$ with rate $R = (p_c - 2 \epsilon) (1 - 1/\beta)$ in the proof.
WLOG, assume $n(p_c - 2\epsilon)$, $n(p_c - \epsilon)$, and $\log_2 n$ are integers as $n \to \infty$.
The encoding scheme, decoding scheme, and decoding error probability are sequentially illustrated in the following.
For the reader's convenience, we present the Hoeffding's inequality \cite{hoeffding1963probability} below.

\begin{lemma}[Hoeffding's inequality \cite{hoeffding1963probability}]\label{lemma: hoeffding inequality}
Consider a coin that shows heads with probability $p$ and tails with probability $1 - p$.
Toss the coin $t$ times and denote $U(t)$ as the number of times the coin comes up heads.
Then, $\mbb{P}(U(t) \leq (p - \eta)t) \leq e^{-2t\eta^2}$ holds for any fixed $\eta > 0$.
\end{lemma}

\emph{Encoding scheme:}
The encoding scheme in Fig. \ref{fig: encoding} is adopted, where we set $k = n(p_c - 2\epsilon)$ and $w = l - \log_2 n = l(1 - 1/\beta)$.
The first encoding step employs a random $(n, k)$ linear code over $\mbb{F}_{2^w}$.
Specifically, let $\mb{G}$ be a generator matrix taken from $\mbb{F}_{2^w}^{k \times n}$ uniformly at random, i.e., each of the $k \times n$ elements of $\mb{G}$ independently takes a value from $\mbb{F}_{2^w}$ with equal probability.
There are $N \triangleq 2^{kw}$ source data matrices $\mb{U}_1, \mb{U}_2, \ldots, \mb{U}_N$ in $\mbb{F}_2^{k \times w}$ (or equivalently in $\mbb{F}_{2^w}^k$).
For any $i \in [N]$, the first encoding step transforms $\mb{U}_i$ into $\mb{V}_i = \mb{G}^{\mr{T}} \mb{U}_i$, where the multiplication operates over $\mbb{F}_{2^w}$ by regarding each row of $\mb{U}_i$ as an element in $\mbb{F}_{2^w}$.
In this case, the following claim holds.

\begin{claim}\label{claim: 2w}
For any $1 \leq i_1 < i_2 \leq N$, a row of $\mb{V}_{i_1}$ equals any given row of $\mb{V}_{i_2}$ with probability $2^{-w}$.
\end{claim}

\begin{IEEEproof}
For $t \in [2]$ and $j_t \in [n]$, denote $\mb{v}_{j_t}$ as the $j_t$-th row of $\mb{V}_{i_t}$.
Then $\mb{v}_{j_t} = \mb{g}_{j_t} \mb{U}_{i_t}$, where $\mb{g}_{j_t}$ denotes the $j_t$-th row in $\mb{G}^{\mr{T}}$.
Since $\mb{U}_{i_1}$ differs from  $\mb{U}_{i_2}$ and both $\mb{g}_{j_1}$ and $\mb{g}_{j_2}$ are taken from $\mbb{F}_{2^w}^{k}$ uniformly at random, the probability of $\mb{v}_{j_1} = \mb{v}_{j_2}$ for any given $j_1, j_2 \in [n]$  is $2^{-w}$ , implying that Claim 1 holds.
\end{IEEEproof}

Next, the second encoding step transforms $\mb{V}_i$ into $\mb{X}_i \in \mbb{F}_2^{n \times l}$ by adding unique addresses.
Finally, let $\{\mb{X}_1, \mb{X}_2, \ldots, \mb{X}_N\}$ form the codebook of $\mc{C}_n$ such that $\mc{C}_n$  has rate $R = \log_2 N/(nl) = kw/(nl) = (p_c - 2 \epsilon) (1 - 1/\beta)$.
The left task it to illustrate the decoding scheme and further show that the decoding error probability tends to zero as $n \to \infty$.

\color{black}

\emph{Decoding scheme:}
After receiving $\mb{Z}$ from the outer channel, if there exists a unique $i \in [N]$ such that $| \mb{X}_i \cap \mb{Z}| \geq n(p_c - \epsilon)$,
then return $\wh{\mb{U}} = \mb{U}_i$ as the decoding result; otherwise, return $\wh{\mb{U}} = ?$ to indicate a decoding failure.
This decoder, which determines the decoding result based on the threshold of $n(p_c - \epsilon)$, can be considered as a simple variant of the ML decoder given by Proposition \ref{prop: ML}.

\emph{Decoding error probability:}
WLOG, assume the source data matrix $\mb{U} = \mb{U}_1$ and the transmitted encoded matrix $\mb{X} = \mb{X}_1$.
For any $i \in [N]$, define $E_i$ as the event of $| \mb{X}_i \cap \mb{Z}| \geq n(p_c - \epsilon)$ and $\bar{E}_i$ as the event of $| \mb{X}_i \cap \mb{Z}| < n(p_c - \epsilon)$.
Then, the decoding error probability is given by
\begin{align}\label{eqn: FER}
    \mbb{P}(\wh{\mb{U}} \neq \mb{U}_1) &= \mbb{P}(\bar{E}_1 \vee E_2 \vee E_3 \vee \cdots \vee E_{N}) \nonumber
    \\&\leq \mbb{P}(\bar{E}_1) + \sum_{i = 2}^{N} \mbb{P}(E_i).
\end{align}

We first bound $\mbb{P}(\bar{E}_1)$.
Let $n_c$ denote the number of correctly transmitted rows of $\mb{X}_1$.
We have $|\mb{X}_1 \cap \mb{Z}| \geq n_c$.
Then,
\begin{align}\label{eqn: FER E1}
    \mbb{P}(\bar{E}_1) &= \mbb{P}(| \mb{X}_1 \cap \mb{Z}| < n(p_c - \epsilon)) \nonumber
    \\&\leq  \mbb{P}(n_c \leq n(p_c - \epsilon)) \nonumber
    \\&\leq e^{-2n\epsilon^2},
\end{align}
where the last inequality is based on Lemma \ref{lemma: hoeffding inequality}.

We now bound $\mbb{P}(E_i)$ for any $2 \leq i \leq N$.
It is important to note the following claim.

\begin{claim}\label{claim: Ei}
$E_i$ happens if and only if (iff) there exists at least a submatrix of $\mb{Z}$, say $\mb{S}$, such that $\mb{S}$ consists of $n(p_c - \epsilon)$ rows of $\mb{Z}$ and $ |\mb{X}_i \cap \mb{S}| = n(p_c - \epsilon)$.
\end{claim}

Claim \ref{claim: Ei} motivates us to define a class of events $E_{i,j}$ related to the set intersection of $\mb{X}_i$ and the submatrices of $\mb{Z}$ , and finally bound $\mbb{P}(E_i)$ by the summation of $\mbb{P}(E_{i,j})$.
Specifically, recall that $\mb{Z}$  is considered as having $n$ rows with any erased row represented by a `?'.
Thus, there are $S \triangleq \binom{n}{n(p_c - \epsilon)}$ many ways to select $n(p_c - \epsilon)$ rows from $\mb{Z}$ to form submatrices.
Denote these submatrices  by $\mb{S}_j, j \in [S]$.
Define  $E_{i,j}$ as the event of $ |\mb{X}_i \cap \mb{S}_j| = n(p_c - \epsilon)$.
Claim \ref{claim: Ei} yields
\begin{align*}
    \mbb{P}(E_i)  = \mbb{P}(\vee_{j = 1}^{S} E_{i,j}) \leq \sum_{j = 1}^{S}  \mbb{P}(E_{i,j}).
\end{align*}
Clearly, $E_{i,j}$ happens iff for all $h \in [n(p_c - \epsilon)]$, $\mb{S}_j$'s $h$-th row  $\mb{s}_{j, h}$ has a unique address $a_h \in [n]$ and $\mb{s}_{j, h}$ equals $\mb{X}_i$'s $a_h$-th row $\mb{x}_{i, a_h}$.
However, even if $\mb{s}_{j, h}$ has a unique address $a_h \in [n]$, $\mb{s}_{j, h}$ only equals $\mb{x}_{i, a_h}$ with a vanished probability of $2^{-w}$ by Claim \ref{claim: 2w}.
Thus, $\mbb{P}(E_{i,j}) \leq 2^{-w n(p_c - \epsilon)}$, leading to
\begin{align}\label{eqn: FER Ei}
    \mbb{P}(E_i)  &\leq \sum_{j = 1}^{S}  \mbb{P}(E_{i,j}) \nonumber
    \\&\leq  S \cdot 2^{-w n(p_c - \epsilon)}\nonumber
    \\&\leq 2^{n-w n(p_c - \epsilon)}.
\end{align}

\color{black}

At this point, substituting \eqref{eqn: FER E1} and \eqref{eqn: FER Ei} into \eqref{eqn: FER} results in
\begin{align}
    \mbb{P}(\wh{\mb{U}} \neq \mb{U}_1)  &\leq  \mbb{P}(\bar{E}_1) + \sum_{i = 2}^{N} \mbb{P}(E_i) \nonumber
    \\&\leq e^{-2n\epsilon^2} + (N - 1) 2^{n-w n(p_c - \epsilon)} \nonumber
    \\&\leq e^{-2n\epsilon^2} + 2^{kw + n - w n(p_c - \epsilon)} \nonumber
    \\&=  e^{-2n\epsilon^2} + 2^{n(p_c - 2\epsilon)w + n - w n(p_c - \epsilon)} \nonumber
    \\&=  e^{-2n\epsilon^2} + 2^{n -  n \epsilon w} \nonumber
    \\&= e^{-2n\epsilon^2} + 2^{n - n \epsilon (\beta - 1)\log n} \nonumber
    \\&\to 0,
\end{align}
where the last step is because $\epsilon$ is fixed, $\beta > 1$, and $n \to \infty$.
This completes the proof.

\subsection{Proof of Lemma \ref{lemma: converse}}

Note that
\begin{equation}\label{eqn: I + H = H}
I(\mb{U}; \widehat{\mb{U}}) + H(\mb{U}|\widehat{\mb{U}}) = H(\mb{U}) = kw = nlR,
\end{equation}
where $I$ is the mutual information function and $H$ is the entropy function.
As $\mb{U}\rightarrow\mb{X}\rightarrow\mb{Z}\rightarrow\widehat{\mb{U}}$ forms a Markov chain, $I(\mb{U}; \widehat{\mb{U}}) \leq I(\mb{X};\mb{Z})$ by the data processing inequality \cite{cover1999elements}.
Meanwhile, according to Fano's inequality \cite{cover1999elements}, we have $H(\mb{U}|\widehat{\mb{U}}) \leq 1 + \mathbb{P}(\mb{U}\neq\widehat{\mb{U}}) \log_2(|\mc{U}|)$, where $\mc{U}$ is the set that $\mb{U}$ takes values from and $|\mc{U}| = 2^{nlR}$ in this paper.
Therefore, \eqref{eqn: I + H = H} becomes
 \begin{align}\label{eqn: Fano}
nlR & = I(\mb{U}; \widehat{\mb{U}}) + H(\mb{U}|\widehat{\mb{U}}) \nonumber\\
&\leq I(\mb{X};\mb{Z})+1+ \mathbb{P}(\mb{U}\neq\widehat{\mb{U}}) n l R,
\end{align}
Since we consider $\mathbb{P}(\mb{U}\neq\widehat{\mb{U}}) \to 0$ as $n \to \infty$ in \eqref{eqn: Fano}, to prove Lemma \ref{lemma: converse}, it suffices to show
\begin{equation}\label{eqn: IXZ}
I(\mb{X}; \mb{Z}) \leq (p_c + \delta) (1 - 1/\beta) nl + o(nl).
\end{equation}
To this end, we partition $\mb{Z}$ into two submatrices $\mb{S}$ and $\bar{\mb{S}}$ which consist of the rows corrupted by uniformly distributed random symbol substitution errors in $\mb{Z}$ and the remaining rows, respectively.
Note that the order of the rows in $\bar{\mb{S}}$ and $\mb{S}$ does not matter.
We have
\begin{equation}\label{eqn: I X S + S}
I(\mb{X}; \mb{Z}) = I(\mb{X}; \bar{\mb{S}}) + I(\mb{X}; \mb{S} | \bar{\mb{S}}).
\end{equation}

We first bound $I(\mb{X}; \bar{\mb{S}})$ in \eqref{eqn: I X S + S}.
Since $\bar{\mb{S}}$ only contains correctly received rows and erased rows, $\bar{\mb{S}}$ can be regarded as the output of the noise-free shuffling-sampling channel \cite{shomorony2021DNA}, where each row of $\mb{X}$ is either sampled once with probability $p_c$ or never sampled with probability $1 - p_c$.
Therefore, using a similar proof as \cite[Section III-B]{shomorony2021DNA} yields
\begin{equation}\label{eqn: I X S bar}
I(\mb{X}; \bar{\mb{S}}) \leq  H(\bar{\mb{S}}) \leq (p_c + \delta) (1 - 1/\beta) nl + o(nl).
\end{equation}

We now bound $I(\mb{X};\mb{S}|\bar{\mb{S}})$  in \eqref{eqn: I X S + S} based on
\begin{equation}\label{eqn: I X, S | S}
I(\mb{X};\mb{S}|\bar{\mb{S}}) = H(\mb{S}|\bar{\mb{S}})-H(\mb{S}|\mb{X},\bar{\mb{S}}).
\end{equation}
Denote $J$ as the number of rows in $\mb{S}$.
We have
\begin{align}\label{eqn: H S | S}
H(\mb{S}|\bar{\mb{S}}) &\overset{(i)}{=} H(\mb{S}|\bar{\mb{S}}, J)\nonumber\\
 &\leq H(\mb{S}|J)\nonumber\\
 &\overset{(ii)}{\leq} \sum_{j=1}^n \mbb{P} (J=j) jl \nonumber\\
 &=p_s n l,
\end{align}
where $(i)$ follows the fact that $J$ is specified by $\bar{\mb{S}}$, and
$(ii)$ holds since  given $J=j$, $\mb{S}$ consists of ${jl}$ bits.
Moreover, for any $i \in [J]$, suppose the $i$-th row $\mb{s}_i$ of $\mb{S}$ came from $\mb{x}_{f_i}$.
Let $\mb{f} \triangleq (f_1, f_2, \ldots, f_J)$.
We have
\begin{align}\label{eqn: H S | X S}
H(\mb{S}|\mb{X},\bar{\mb{S}}) &= H(\mb{S}|\mb{X},\bar{\mb{S}}, J)\nonumber\\
 &\geq H(\mb{S}|\mb{X},\bar{\mb{S}}, J, \mb{f})\nonumber\\
 &\overset{(iii)}{=} \sum_{j=1}^n \mbb{P} (J=j) \sum_{i = 1}^{j} H(\mb{s}_i | \mb{x}_{f_i})\nonumber\\
 &=p_s n \log_2(2^l-1),
\end{align}
where $(iii)$ holds because given $(\mb{X}, J, \mb{f})$, for any $i \in [J]$, the $i$-th row $\mb{s}_i$ of $\mb{S}$ only relates to $\mb{x}_{f_i}$, and  $\mb{s}_i$ has $2^{l}-1$ equiprobable choices according to \eqref{eqn: channel-1}.
By applying \eqref{eqn: H S | S} and \eqref{eqn: H S | X S} to \eqref{eqn: I X, S | S}, we have
\begin{equation}\label{eqn: I X, S | S onl}
I(\mb{X};\mb{S}|\bar{\mb{S}}) \leq -\log_2(1 - 2^{-l}) = o(1).
\end{equation}

Finally, by substituting \eqref{eqn: I X S bar} and \eqref{eqn: I X, S | S onl} into \eqref{eqn: I X S + S} yielding \eqref{eqn: IXZ}, the proof is completed.

\section{Practically Efficient Coding Schemes}\label{section: Decoding Scheme}

To achieve a good trade-off between efficiency and complexity, this section adopts the encoding scheme  in Fig. \ref{fig: encoding} with its first encoding step employing binary $(n, k)$ linear ECCs to independently encode each column of the source data matrix.
To perform decoding efficiently, we first derive the soft information $\mb{M} = (m_{i, j})_{1 \leq i \leq n, 1 \leq j \leq w}$ and hard information $\wh{\mb{M}} = (\wh{m}_{i, j})_{1 \leq i \leq n, 1 \leq j \leq w}$ of  each  bit in $\mb{V}$.
This allows us to decode each column of $\mb{V}$ independently, leading to the  independent decoding scheme.
Next, in view that it requires the successful decoding of all columns to fully recover $\mb{V}$ (as well as $\mb{U}$ and $\mb{X}$), we propose an enhanced joint decoding scheme, which measures  the reliability of rows of $\mb{Z}$ based on the independent decoding result  and takes the most reliable rows for a further step of decoding.

Given any $i \in [n]$ and $j \in [w]$, we begin with  the computation of the soft information $m_{i, j}$ of $v_{i,j} \in \mb{V}$.
Conventionally, we should define $m_{i, j} \triangleq \mathbb{P}(v_{i,j} = 0 \mid \mb{Z}) / \mathbb{P}(v_{i,j} = 1 \mid \mb{Z})$.
In fact, $m_{i, j}$ is only related to $\mb{y}_i$.
 But the problem is that we do not know $\mb{y}_i$ corresponds to which row in $\mb{Z}$ due to the random permutation of channel-2, making it hard to compute $m_{i, j}$ exactly.
Note from \eqref{eqn: channel-1} that, if $\mb{y}_i$ is erased or corrupted by uniformly distributed random symbol substitution errors such that its address is not $i$ any more, no (or at most negligible) information about $v_{i, j}$ can be inferred from  $\mb{y}_i$ even if we can figure out which row in $\mb{Z}$ corresponds to $\mb{y}_i$.
Therefore, it is reasonable to  reduce the computation of $m_{i, j}$ from $\mb{Z}$ to the rows  with address $i$.
Accordingly, we define
\begin{equation}\label{eqn: mij definition}
    m_{i, j} \triangleq \frac{\mathbb{P}(v_{i,j} = 0 \mid t, t_0) }{ \mathbb{P}(v_{i,j} = 1 \mid t, t_0)},
\end{equation}
where $t$ denotes  the number of rows in $\mb{Z}$ with address $i$, among which there are $t_0$ rows with the $j$-th bit being  $0$.
Based on \eqref{eqn: mij definition}, $m_{i,j}$ can be computed by  the following proposition.

\begin{proposition}\label{proposition: mij}

\begin{align}\label{eqn: mij}
m_{i,j} = &\big(2t_0(1 - q)(p_1 + p_4) + (n - t)q(p_2 + p_3) \nonumber
    \\&\quad  + 2(t - t_0)(1 - q)p_5 \big) \nonumber
    \\& \big/ \big( 2(t - t_0)(1 - q)(p_1 + p_4) \nonumber
    \\&\quad  + (n - t)q(p_2 + p_3) + 2t_0(1 - q)p_5 \big),
%m_{i,j} = \frac{ 2t_0(1 - q)(p_1 + p_4) + (n - t)q(p_2 + p_3)
%      + 2(t - t_0)(1 - q)p_5 }{
%     2(t - t_0)(1 - q)(p_1 + p_4)  + (n - t)q(p_2 + p_3) + 2t_0(1 - q)p_5 },
\end{align}
where
\begin{equation}\label{eqn: q}
    q \triangleq p_s 2^{l - a} / (2^l - 1)
\end{equation}
and
\begin{equation}\label{eqn: p1-p5}
    p_b \triangleq
    \begin{cases}
        p_c, &b = 1,
        \\p_e, &b = 2,
        \\p_s (2^l - 2^{l - a}) / (2^l - 1), &b = 3,
        \\p_s (2^{l - a - 1} - 1)/(2^l - 1), &b = 4,
        \\p_s 2^{l - a - 1}/(2^l - 1), &b = 5.
    \end{cases}
\end{equation}

\end{proposition}

\begin{IEEEproof}
See Appendix \ref{appendix: proof of mij}.
\end{IEEEproof}

Based on Proposition \ref{proposition: mij}, it is natural to define the hard information of $v_{i,j}$ by:
\begin{equation}\label{eqn: hat mij mij}
    \wh{m}_{i, j} \triangleq
    \begin{cases}
        ?, & m_{i, j} = 1,\\
        0, & m_{i, j} > 1,\\
        1, & m_{i, j} < 1.
    \end{cases}
\end{equation}
The following corollary simplifies \eqref{eqn: hat mij mij}.

\begin{corollary}
\begin{equation}\label{eqn: hat mij t t0}
    \wh{m}_{i, j} =
    \begin{cases}
        ?, & t_0 = t/2,\\
        0, & t_0 > t/2,\\
        1, & t_0 < t/2.
    \end{cases}
\end{equation}
\end{corollary}

\begin{proof}
By Proposition \ref{proposition: mij}, we have
\begin{align}
m_{i, j} \geq 1 &\iff (2t_0 - t)(p_1 + p_4 - p_5) \geq 0 \nonumber
\\&\iff (2t_0 - t)(p_c - p_s/(2^l - 1)) \geq 0 \nonumber
\\&\iff 2t_0 - t \geq 0,
\end{align}
where the last step is due to  \eqref {eqn: pc > ps} and the equality holds iff $t_0 = t/2$.
\end{proof}

Eqn. \eqref{eqn: hat mij t t0} is quite intuitive since it coincides with the majority decoding result.
Hence, it somehow verifies the correctness of \eqref{eqn: mij}.
Given the soft information in \eqref{eqn: mij} or hard information in \eqref{eqn: hat mij t t0}, it is able to perform decoding independently for each column of $\mb{V}$, leading to the independent decoding scheme, as shown in Algorithm \ref{algorithm: independent}.
However, Algorithm \ref{algorithm: independent} can fully recover $\mb{V}$ (as well as $\mb{U}$) only if the decoding for each column succeeds.
This condition is generally too strong to achieve a desired decoding performance, as illustrated by the following example.

\begin{algorithm}[t!]
    \caption{Independent decoding scheme}
    \label{algorithm: independent}
    \begin{algorithmic}[1]
        \REQUIRE $\mb{Z} = [(\mb{z}_i^{\mr{T}})_{1 \leq i \leq n}]^{\mr{T}} = ({z}_{i,j})_{1 \leq i \leq n, 1 \leq j \leq l}$.
        \ENSURE $\wh{\mb{U}}$.
        \STATE Calculate the soft/hard information of $v_{i,j}, i \in[n], j \in [w]$ based on \eqref{eqn: mij}/\eqref{eqn: hat mij t t0}.
        \STATE  Decode each column of $\mb{V}$ independently and denote the result by $\tilde{\mb{V}} = (\tilde{v}_{i,j})_{1 \leq i \leq n, 1 \leq j \leq w}$, where $\tilde{v}_{i,j} \in \mbb{F}_2$ if decoding the $j$-th column of $\mb{V}$ gives a codeword and $\tilde{v}_{i,j} = ?$ otherwise.
        \STATE If $\tilde{\mb{V}}$ contains `?', set $\wh{\mb{U}} = ?$, otherwise retrieve $\wh{\mb{U}}$ from $\tilde{\mb{V}}$.
        \RETURN $\wh{\mb{U}}$.
    \end{algorithmic}
\end{algorithm}

\begin{example}\label{example: decoding}
Continued with Example \ref{example: encoding}.
Suppose the $\mb{X}$ in \eqref{eqn: X eg encoding} is transmitted over channel-1 and the output is
\begin{equation}\label{eqn: Y eg decoding}
    \mb{Y} =
    \left[
        \begin{array} {c c c c | c c c}
            0 & 0 & \ul{0} & \ul{0} & 0 & \ul{1} & \ul{0}\\
            0 & 1 & 0 & 1 & 0 & 1 & 0\\\hline
            \ul{1} & \ul{1} & 1 & 1 & 0 & 1 & 1\\
            0 & 1 & 1 & 0 & 1 & 0 & 0\\
            0 & 1 & 1 & 0 & 1 & 0 & 1\\
            0 & 1 & 0 & 1 & 1 & 1 & 0\\
        \end{array}
    \right],
\end{equation}
where the underlined bits are flipped by channel-1.
That is, $\mb{y}_1$ is corrupted by substitution errors and its address changes to $\mb{y}_2$'s address; $\mb{y}_3$ is also corrupted by substitution errors but its address does not change.
Assume $\mb{Z} = \mb{Y}$ for convenience (but the decoder does not know the correspondence between the rows of $\mb{Z}$ and $\mb{Y}$).

We compute the hard information of $\mb{V}$ by \eqref{eqn: hat mij t t0}, leading to
\begin{equation}\label{eqn: M eg decoding}
    \wh{\mb{M}} =
    \left[
        \begin{array} {c c c c}
            ? & ? & ? & ?\\
            0 & ? & 0 & ?\\\hline
            \tcr{\mb{1}} & \tcr{\mb{1}} & 1 & 1\\
            0 & 1 & 1 & 0\\
            0 & 1 & 1 & 0\\
            0 & 1 & 0 & 1\\
        \end{array}
    \right],
\end{equation}
where the bold red bits are different from the corresponding bits of $\mb{V}$ in \eqref{eqn: V eg encoding}.
For examples, $\wh{m}_{1,1} = ?$ since $\mb{Z}$ does not have a row with address 1; $\wh{m}_{2,2} = ?$ since $\mb{Z}$ has $t = 2$ rows with address 2 and their second bits differ from each other, i.e., $t_0 = 1$.
It is natural to decode each column of $\wh{\mb{M}}$ as the nearest codeword (in terms of Hamming distance) under the parity check matrix $\mb{H}$ in \eqref{eqn: H}, leading to
\begin{equation}\label{eqn: hat V eg decoding}
    \tilde{\mb{V}} =
    \left[
        \begin{array} {c c c c}
            0 & ? & 1 & 1\\
            0 & ? & 0 & 1\\\hline
            0 & ? & 1 & 1\\
            0 & ? & 1 & 0\\
            0 & ? & 1 & 0\\
            0 & ? & 0 & 1\\
        \end{array}
    \right],
\end{equation}
where the second column corresponds to a decoding failure since both $(0,1,0,1,1,1)^{\mr{T}}$ and $(1,0,1,1,1,0)^{\mr{T}}$ are the nearest codewords of the second column of $\wh{\mb{M}}$.
Thus, Algorithm \ref{algorithm: independent} fails to fully recover $\mb{V}$.
\end{example}

We find that in Example \ref{example: decoding}, all the known bits in $\tilde{\mb{V}}$ are correct.
In a general situation, the known bits in $\tilde{\mb{V}}$ should also be correct with a high probability.
We thus assume
\begin{equation}\label{eqn: prob hat vij = vij}
    \mbb{P}(\tilde{v}_{i, j} = v_{i,j}) > 1/2, \quad \forall \tilde{v}_{i, j} \in \tilde{\mb{V}} \text{~and~} \tilde{v}_{i, j}  \in \mbb{F}_2.
\end{equation}
%On the other hand, each non-erased row of $\mb{Z}$ is either correct or corrupted by random substitution errors according to \eqref{eqn: channel-1}.
Consequently, for a non-erased correct (resp. incorrect) row in $\mb{Z}$ with address $i \in [n]$, it generally has a relatively small (resp. large) distance from $\tilde{\mb{v}}_i$.
This provides a way to measure the reliability of rows of $\mb{Z}$.
We can then perform a further step of decoding over the most reliable rows of $\mb{Z}$ like under the erasure channel,  leading to the joint decoding scheme, as shown by Algorithm \ref{algorithm: joint}.
In the following, we give more explanations (including Example \ref{example: SGE for Hamming code}) for Algorithm \ref{algorithm: joint}.

\begin{algorithm}[t!]
    \caption{Joint decoding scheme}
    \label{algorithm: joint}
    \begin{algorithmic}[1]
        \REQUIRE $\mb{Z} = [(\mb{z}_i^{\mr{T}})_{1 \leq i \leq n}]^{\mr{T}} = ({z}_{i,j})_{1 \leq i \leq n, 1 \leq j \leq l}$.
        \ENSURE $\wh{\mb{U}}$.
        \STATE Execute Lines 1 and 2 of Algorithm \ref{algorithm: independent}.
        \STATE  For each $i \in [n]$, if $\mb{z}_i = ?$, set $d_i = w$; otherwise, set $d_i = |\{j \in [w]: z_{i, j} \neq \tilde{v}_{i',j}\}|$, where $i'$ is the address of $\mb{z}_i$.
        \STATE Rearrange $\mb{z}_i$ in ascending order with respect to $d_i, \forall i \in [n]$. $\%$As a result, $d_i \leq d_{i'}, \forall 1 \leq i < i' \leq n$, indicating that $\mb{z}_i$ has higher reliability than $\mb{z}_{i'}$.
        \STATE  Find the smallest $n' \in [n]$ such that by viewing $\mc{Z}_{n'} \triangleq \{\mb{z}_{i}: i \in [n']\}$ as correct and $\mb{Z} \setminus \mc{Z}_{n'} = \{\mb{z}_{i}: i \in [n] \setminus [n']\}$  as erased, the decoding   over $\mc{Z}_{n'}$ gives a unique estimation $\wh{\mb{U}}$ of $\mb{U}$. If no such an $n'$ exists, set $\wh{\mb{U}} = ?$.
        \RETURN $\wh{\mb{U}}$.
    \end{algorithmic}
\end{algorithm}

\black

\begin{itemize}
\item   Line 2 is to count the number $d_i$ of different positions (Hamming distance) between $\mb{z}_i$ and $\tilde{\mb{v}}_{i'}$, where $i'$ is the address of $\mb{z}_i$.
    According to \eqref{eqn: prob hat vij = vij}, larger $d_i$ implies a higher reliability of $\mb{z}_i$.
\item   Lines 3 and 4 are to take the most reliable rows $\mc{Z}_{n'}$ for a further step of decoding, where these rows are viewed as correct and the other rows of $\mb{Z}$ are viewed as erasure errors.
    Thus, the decoding actually works like under the erasure channel, which essentially is to solve linear systems when decoding linear ECCs.
    As our simulations are based on LDPC codes, the structured gaussian elimination (SGE) \cite{he2020disjoint, shokrollahi2005systems, odlyzko1984discrete}, also called inactivation decoding, is recommended since it  works quite efficiently for solving large sparse linear systems.
    The SGE leads to $\wh{\mb{U}} = \mb{U}$ if each row in $\mc{Z}_{n'}$ is correct and $n'$ is sufficiently large to uniquely determine $\wh{\mb{U}}$ (see Example \ref{example: SGE for Hamming code}).
\item   Algorithm \ref{algorithm: joint} definitely has a better chance to successfully recover $\mb{U}$ than Algorithm \ref{algorithm: independent}, since it succeeds if the latter succeeds and it may still succeed otherwise.
\end{itemize}

\begin{example}\label{example: SGE for Hamming code}
%example for parity-check matrix
Continued with Example \ref{example: decoding}.
Recall that $\mb{Z} = \mb{Y}$ given by  \eqref{eqn: Y eg decoding} before executing Line 3 of  Algorithm \ref{algorithm: joint}.
As a result, Line 1 of  Algorithm \ref{algorithm: joint} obtains $\tilde{\mb{V}}$ in \eqref{eqn: hat V eg decoding}.
Then, Line 2 gives $(d_i)_{1 \leq i \leq n} = (2, 1, 2, 1, 1, 1)$.
Next, Line 3 rearranges the rows of $\mb{Z}$, say $\mb{Z} = [\mb{z}_1^{\mr{T}}, \mb{z}_2^{\mr{T}}, \mb{z}_3^{\mr{T}}, \mb{z}_4^{\mr{T}}, \mb{z}_5^{\mr{T}}, \mb{z}_6^{\mr{T}}]^{\mr{T}} = [\mb{y}_2^{\mr{T}}, \mb{y}_4^{\mr{T}}, \mb{y}_5^{\mr{T}}, \mb{y}_6^{\mr{T}}, \mb{y}_1^{\mr{T}}, \mb{y}_3^{\mr{T}}]^{\mr{T}}$, where $\mb{y}_i, i\in[6]$ is the $i$-th row of the $\mb{Y}$ in \eqref{eqn: Y eg decoding}.

In Line 4, for a given $n' \in [n]$, $\mc{Z}_{n'}$ consists of the first $n'$ rows in $\mb{Z}$.
Recall that $\mb{H}$ given by \eqref{eqn: H} is the parity-check matrix.
Our task is to find the smallest $n'$ such that there exists a unique $\wh{\mb{V}} = [(\wh{\mb{v}}_i^{\mr{T}})_{1 \leq i \leq n}]^{\mr{T}}$ (then $\wh{\mb{U}}$ is uniquely determined) satisfying  $\mb{H} \wh{\mb{V}} = \mb{0}$ and $\mc{Z}_{n'}$ is a part of $\wh{\mb{V}}$.
Here we say $\mc{Z}_{n'}$ is a part of $\wh{\mb{V}}$ if for any $i \in [n']$ such that $\mb{z}_i$  has a valid address $a_i \in [n]$, $\wh{\mb{v}}_{a_i}$ equals the data of $\mb{z}_i$.

More specifically, we first try $n' = 1$, yielding $\mc{Z}_{1} = \{\mb{z}_1\} = \{\mb{y}_2\}$, where $\mb{z}_1$ has address 2.
Thus, we needs to solve $\mb{H} \wh{\mb{V}} = \mb{0}$ with  $\wh{\mb{v}}_2$ being known and $\wh{\mb{v}}_i, i \in \{1, 3, 4, 5, 6\}$ being unknowns.
Since  the 1st, 3rd, 4th, 5th, and 6th columns of $\mb{H}$ are linearly dependent, these unknowns cannot be uniquely determined.
We need to increase $n'$.

We next try $n' = 2$, yielding $\mc{Z}_{2} = \{\mb{z}_1, \mb{z}_2\} = \{\mb{y}_2, \mb{y}_4\}$, where $\mb{z}_2$ has address 4.
We needs to solve $\mb{H} \wh{\mb{V}} = \mb{0}$ with  $\{\wh{\mb{v}}_2, \wh{\mb{v}}_4\}$ being known and $\wh{\mb{v}}_i, i \in \{1, 3, 5, 6\}$ being unknowns.
Since  the 1st, 3rd, 5th, and 6th columns of $\mb{H}$ are linearly independent, these unknowns can be uniquely determined.
Moreover, since $\wh{\mb{v}}_2$ and $\wh{\mb{v}}_4$ are correct, the decoding succeeds, i.e., $\mb{V}$ as well as $\mb{U}$ are fully recovered.
\end{example}

To end this section, we discuss the (computational) complexities of both Algorithms \ref{algorithm: independent} and \ref{algorithm: joint}.
For Algorithm \ref{algorithm: independent}, it decodes each column of $\mb{V}$ independently.
Thus, its complexity is $O(w c_1)$, where $w$ is the number of data columns and $c_1$ denotes the complexity for decoding one column.
The value of $c_1$ depends on the underlying ECC as well as the decoding algorithm.
On the other hand, Algorithm \ref{algorithm: joint} has complexities $O(w c_1)$, $O(w n)$, and  $O(n \log_2 n)$ for its Lines 1--3, respectively.
The complexity of Line 4 also depends on the underlying ECC and decoding algorithm.
However, since Line 4 is to correct erasures, its complexity is less than or equal to the complexity $O(w c_1)$ of Algorithm \ref{algorithm: independent}.
As a result, the total complexity of Algorithm \ref{algorithm: joint} is $O(w c_1 + w n + n \log_2 n)$, which generally equals $O(w c_1)$ since $c_1 \geq n$ must hold and $w \geq \log_2 n$ is required in practice to have a good code rate.
That is, Algorithms \ref{algorithm: independent} and \ref{algorithm: joint} have the same order of complexity.

For example, suppose the underlying ECC is an LDPC code which has a sparse parity-check matrix of a total number $\theta$ of non-zero entries.
Then, Algorithm \ref{algorithm: independent} has complexity $O(w c_1) = O(w \theta t_{max})$ under message-passing algorithms with a max number  $t_{max}$ of iterations.
Moreover, Line 4 of Algorithm  \ref{algorithm: joint} generally has complexity $O(w \theta)$ \cite{he2020disjoint} such that the total complexity of Algorithm  \ref{algorithm: joint} is  $O(w \theta t_{max}  + w n + n \log_2 n) = O(w \theta t_{max})$ in practical situations.
Our simulations confirm that Algorithm  \ref{algorithm: joint} is only slightly slower than Algorithm  \ref{algorithm: independent}.

\section{Simulation Results}\label{section: simulation}

In this section, we consider using binary LDPC code as the outer code since it works very well with soft information.
We evaluate the FERs of both the independent decoding scheme (Algorithm  \ref{algorithm: independent}) and joint decoding scheme (Algorithm  \ref{algorithm: joint}).
A frame error occurs when a test case does not fully recover the data matrix $\mb{U}$.
We collect at least $100$ frame errors at each simulated point.
The independent decoding scheme first computes the soft information by \eqref{eqn: mij}, and then adopts the belief propagation (BP) algorithm \cite{ECC04} with  a maximum of $100$ iterations to recover each column of $\mb{U}$.
Given the independent decoding  result, the joint decoding scheme further executes Lines 2--4 of Algorithm \ref{algorithm: joint} to recover  $\mb{U}$.
Since LDPC codes are characterized by their sparse parity-check matrices, the SGE\cite{he2020disjoint, shokrollahi2005systems, odlyzko1984discrete} is adopted for the decoding in Line 4 of Algorithm  \ref{algorithm: joint}, and the decoding process is  similar to that in Example \ref{example: SGE for Hamming code}.

\begin{figure}[!t]
\centering
\subfigure[$p_s \in \{0, 0.01, 0.05\}$.]{
    \includegraphics[scale = 0.58]{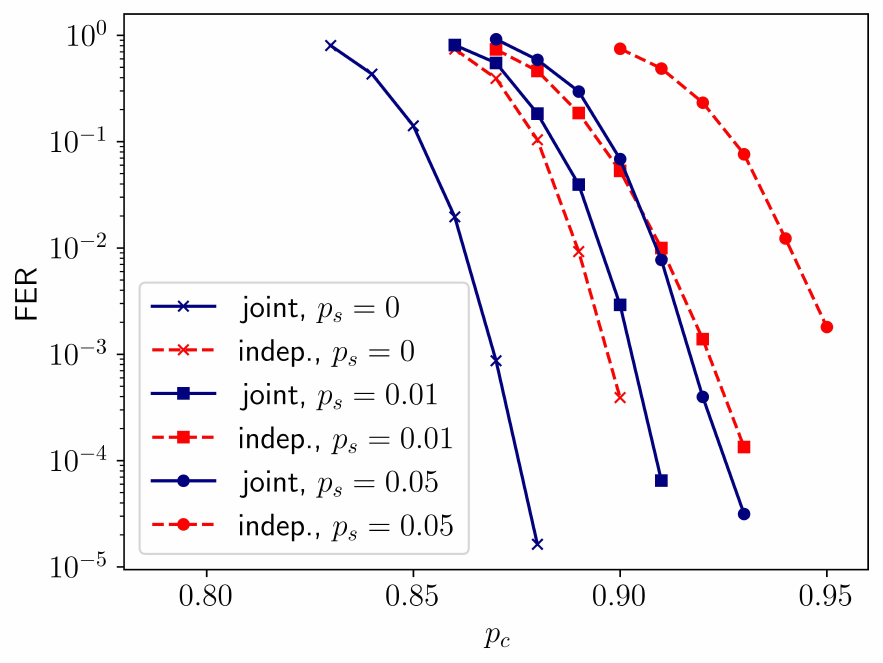}
}

\subfigure[$p_e \in \{0, 0.01, 0.05\}$.]{
    \includegraphics[scale = 0.58]{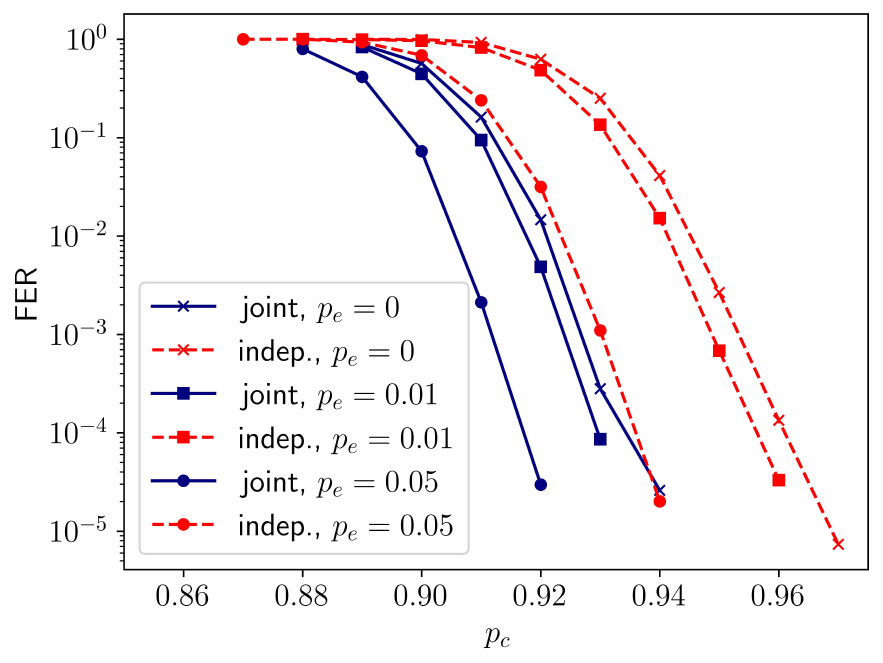}
}
\caption{For $l = 100$, FER performance of both the independent and joint decoding schemes for decoding the $(1296, 1080)$ LDPC code \cite{IEEESTD802_11n}.}
\label{fig: FER l100 ps pe}
\end{figure}

We first consider the $(n = 1296, k = 1080)$ LDPC code \cite{IEEESTD802_11n}.
For fixed  $l = 100$, the simulation results are shown in Fig. \ref{fig: FER l100 ps pe}.
We can see that:
\begin{itemize}
\item   For the same $(l, p_c, p_e, p_s)$, the joint decoding scheme always achieves lower FER  than the independent decoding scheme.
    The difference can exceed $3$ orders of magnitude.
\item   For the same $(l, p_c)$, a decoding scheme has decreasing FER  as $p_e$ increases (or equivalently as $p_s$ decreases).
    It implies that the uniformly distributed random symbol substitution errors are more harmful than the erasure errors, which is inconsistent with  Theorem \ref{theorem: capacity} which indicates that the two error types are of equal harm.
    We believe the main reason leading to this inconsistence is that, Theorem \ref{theorem: capacity} has been proven by using infinite-length random linear codes over $\mbb{F}_{2^w}$ with near ML decoding, while the simulations are performed on thousands-bit long binary LDPC codes with BP decoding.
    This is a common phenomenon in information theory caused by the difference between infinite-length random  coding and the finite-length coding, e.g., see \cite{polyanskiy2010channel}.
    Moreover, it seems very difficult or even impossible to prove Theorem \ref{theorem: capacity} for the case where a random binary linear code is used to  independently encode each column of $\mb{U}$, since in this case  Claim \ref{claim: 2w} is no longer true.
    More specifically, the probability in Claim \ref{claim: 2w} may be as large as $1/2$ when two source data matrices only differ from each other by one bit, making it hard to bound  $\mbb{P}(E_i)$ by a vanished probability in \eqref{eqn: FER Ei} such that the overall decoding error probability may not tend to zero.
\end{itemize}

\begin{figure}[!t]
\centering
\subfigure[$(1296, 1080)$ LDPC code and LT code.]{
    \includegraphics[scale = 0.58]{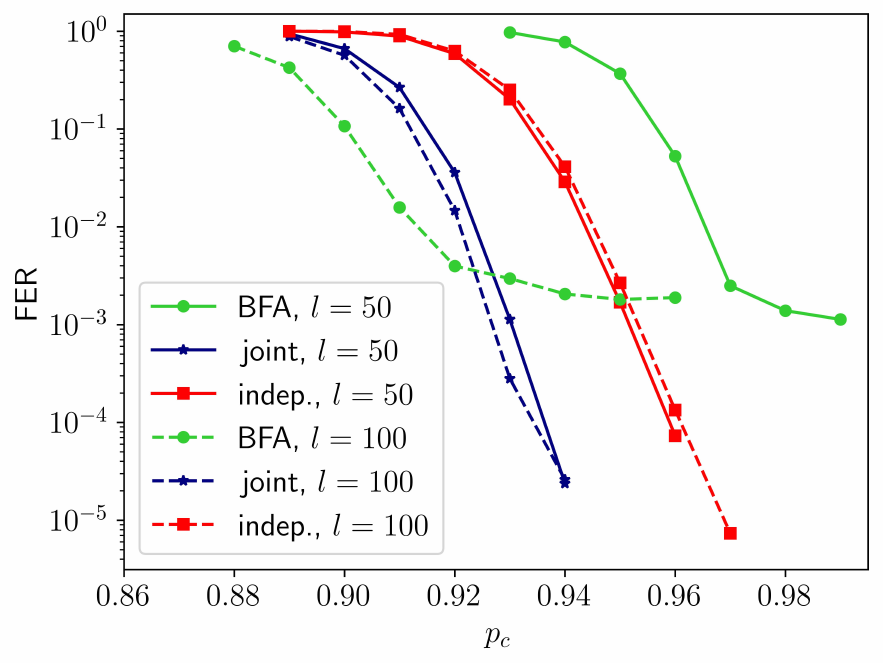}
}

\subfigure[$(2592, 2160)$ LDPC code and LT code.]{
    \includegraphics[scale = 0.58]{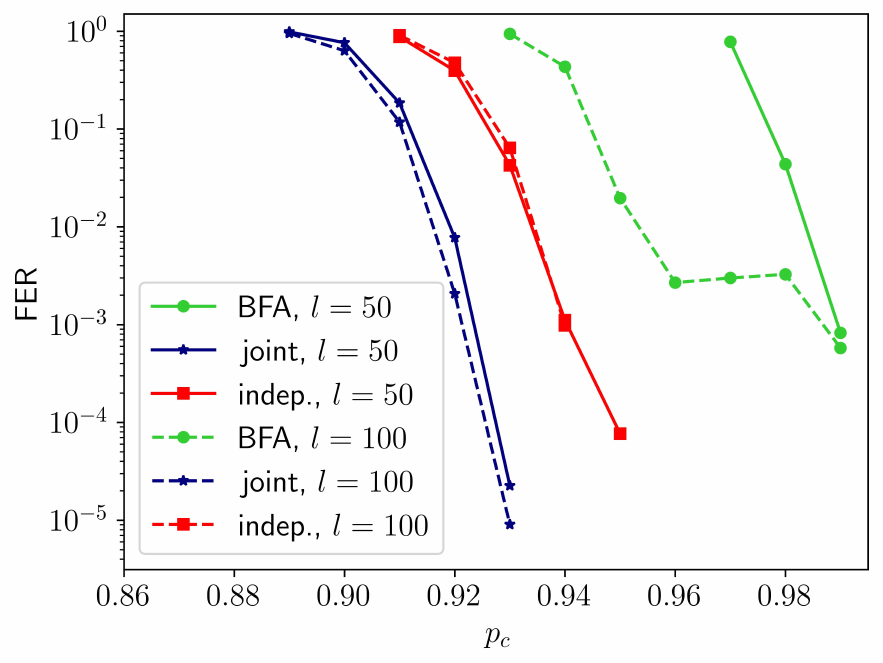}
}
\caption{For $p_e = 0$, FER performance of both the independent and joint decoding schemes for decoding LDPC codes, and FER performance of the BFA for decoding LT codes. }
\label{fig: FER 2592}
\end{figure}

Next, we fix $p_e = 0$ and vary $(n, l)$.
In Fig. \ref{fig: FER 2592}(a), the independent and joint decoding schemes are used for decoding the $(1296, 1080)$ LDPC code \cite{IEEESTD802_11n}.
Meanwhile, in Fig. \ref{fig: FER 2592}(b), they are used for decoding the $(2592, 2160)$ LDPC code, which is constructed by enlarging the lifting size of the $(1296, 1080)$ LDPC code from $54$ to $108$.
The FER of the BFA \cite{he2023basis}  for decoding the Luby Transform (LT) codes \cite{luby2002lt} is presented as a baseline, where the BFA takes the sorted-weight implementation.
For a fair comparison, the LT codes have code length $n = 1296$ and information length $k = 1080$ in Fig. \ref{fig: FER 2592}(a) and have $(n, k) = (2592, 2160)$ in Fig. \ref{fig: FER 2592}(b), and each LT symbol consists of  $a = \lceil \log_2 n \rceil$ address (seed) bits and  $w = l - a$ data bits.
The robust soliton distribution (RSD) with $(\delta, c) = (0.01, 0.02)$ \cite{he2023basis} is chosen to generate LT symbols.
From Fig. \ref{fig: FER 2592}, we can see that:
 \begin{itemize}

\item   For the same $(n, p_c, p_e, p_s)$, as $l$ changes from $50$ to $100$, FER of the independent decoding scheme slightly increases since it needs to correctly recover more columns to fully recover $\mb{U}$;
    FER of the joint decoding scheme obviously decreases since the independent decoding result can provide more information to better measure the reliability of rows of $\mb{Z}$;
    FER of the BFA significantly decreases except for the error floor region.
\item   For the same $(l, p_c, p_e, p_s)$, as $n$ changes from $1296$ to $2592$, FERs of both the independent and joint decoding schemes obviously decrease,  since longer LDPC codes lead to stronger error-correction capability;
    FER of the BFA significantly  increases except for the error floor region.
\end{itemize}

In summary, we should always choose the joint decoding scheme compared to the independent decoding scheme.
In addition, the joint decoding scheme and the BFA can outperform each other in different parameter regions.
Specifically, increasing $n$ and/or $l$ can obviously reduce the FER of the joint decoding scheme, since larger $n$ leads to stronger error-correction capability and larger $l$ can get more information from the independent decoding result.
However, the BFA is very sensitive to $l / n$.
According to \cite{he2023basis} as well as Fig. \ref{fig: FER 2592}, the BFA requires $l/n > p_s$ to have low FER.
Therefore, the joint decoding scheme can generally outperform the BFA for relatively large $n$ and small $l$, e.g., see Fig. \ref{fig: FER 2592}.

\section{Conclusions}\label{section: conclusion}

In this paper, we considered the outer channel  in Fig. \ref{fig: Channel model} for DNA-based data storage.
We first derived the capacity of the outer channel, as stated by Theorem \ref{theorem: capacity}.
It implies that simple index-based coding scheme is optimal and uniformly distributed random symbol substitution errors are only  as harmful as erasure errors.
Next,  we derived the soft and hard information of data bits given by \eqref{eqn: mij} and \eqref{eqn: hat mij t t0}, respectively.
These information was used to decode each column of $\mb{U}$ independently, leading to the independent decoding scheme.
Based on the independent decoding result, the reliability of each row of $\mb{Z}$ can be measured.
Selecting the most reliable rows to recover  $\mb{U}$, similar to the case under the erasure channel, leads to the joint decoding scheme.
Simulations showed that the joint decoding scheme can reduce the frame error rate (FER) by more than 3 orders of magnitude compared to the independent decoding scheme, and the joint decoding scheme and basis-finding algorithm (BFA) \cite{he2023basis} can outperform each other in different parameter regions.

\appendices
\section{Proof of Proposition \ref{proposition: mij}}\label{appendix: proof of mij}

Consider the transmission of $\mb{x}_i$ over channel-1.
$\mb{y}_i$ is the channel output.
We define the following five events:
\begin{itemize}
\item   $E_1$:  $\mb{x}_i$ is correctly transmitted, i.e., $\mb{y}_i = \mb{x}_i$.
\item   $E_2$:  $\mb{x}_i$ is erased, i.e., $\mb{y}_i = ?$.
\item   $E_3$:  $\mb{x}_i$ changes to $\mb{y}_i \in \mbb{F}_2^{l} \setminus \{\mb{x}_i\}$ with a different address.
\item   $E_4$:  $\mb{x}_i$ changes to $\mb{y}_i \in \mbb{F}_2^{l} \setminus \{\mb{x}_i\}$ with the same address and the same $j$-th bit.
\item   $E_5$:  $\mb{x}_i$ changes to $\mb{y}_i \in \mbb{F}_2^{l} \setminus \{\mb{x}_i\}$ with the same address and a different $j$-th bit.
\end{itemize}
It is easy to figure out that $\mbb{P}(E_b) = p_b$, where $b \in [5]$ and $p_b$ is given by \eqref{eqn: p1-p5}.
As a result, we have
\begin{align}\label{align: lr expression}
    m_{i,j} = &\frac{\mathbb{P}(v_{i,j} = 0 \mid t, t_0)}{\mathbb{P}(v_{i,j} = 1 \mid t, t_0)} \nonumber
            \\=&\frac{\mathbb{P}(t, t_0 \mid v_{i,j} = 0)}{\mathbb{P}(t, t_0 \mid v_{i,j} = 1)} \nonumber
            \\=&\frac{\sum_{b \in [5]}{\mathbb{P}(t, t_0 \mid v_{i,j} = 0, E_b})\mathbb{P}(E_b)} {\sum_{b \in [5]}{\mathbb{P}(t, t_0 \mid v_{i,j} = 1, E_b})\mathbb{P}(E_b)} \nonumber
            \\=&\frac{\sum_{b \in [5]}{\mathbb{P}(t, t_0 \mid v_{i,j} = 0, E_b})p_b} {\sum_{b \in [5]}{\mathbb{P}(t, t_0 \mid v_{i,j} = 1, E_b})p_b}.
\end{align}

To compute each $\mathbb{P}(t, t_0 \mid v_{i,j}, E_b)$ in \eqref{align: lr expression}, it needs the probability of that the address of $\mb{y}_{i'}$ is $i$ for a specific $i' \in [n] \setminus \{i\}$ (i.e., the address of $\mb{x}_{i'}$ changes into $i$ after transmission over channel-1).
According to \eqref{eqn: channel-1}, this probability is  $q$ given by \eqref{eqn: q}.
If we further require the $j$-th bit of $\mb{y}_{i'}$ being  $0$, the probability reduces to $q/2$.
For any $0 \leq t'_0 \leq t' < n$, we define the following event:
\begin{itemize}
\item   $E_{t', t'_0}$: the addresses of exact $t'$ rows in $\mb{Y} \setminus \{\mb{y}_i\}$ are $i$, and for exact $t'_0$ out of the $t'$ rows, the $j$-th bit equals to $0$.
\end{itemize}
According to \eqref{eqn: channel-1} and \eqref{eqn: q}, we have
\begin{align}\label{eqn: P(Ett0)}
   \mbb{P} (E_{t', t'_0}) &= \frac{(n-1)!}{t'_0!(t' - t'_0)!(n-1-t')!}\nonumber
   \\&\quad\quad \times  \left({q}/{2}\right)^{t'_0} \left({q}/{2}\right)^{t' - t'_0} (1 - q)^{n - 1 - t'}\nonumber
   \\&= \frac{(n-1)! \left({q}/{2}\right)^{t'} (1 - q)^{n - 1 - t'}}{t'_0!(t' - t'_0)!(n-1-t')! }.
\end{align}

At this point, we can use \eqref{eqn: P(Ett0)} to compute each $\mathbb{P}(t, t_0 \mid v_{i,j}, E_b)$ in \eqref{align: lr expression}.
Specifically, the following results hold:

\begin{align}\label{align: combination 1}
    \mathbb{P}(t, t_0 \mid v_{i,j} = 0, E_1) \nonumber
    = &\mathbb{P}(t, t_0 \mid v_{i,j} = 0, E_4) \nonumber
    \\ = &\mathbb{P}(t, t_0 \mid v_{i,j} = 1, E_5) \nonumber
    \\ = & \mbb{P}(E_{t - 1, t_0 - 1}),
\end{align}
\begin{align}\label{align: combination 2}
    \mathbb{P}(t, t_0 \mid v_{i,j} = 0, E_2)
    = &\mathbb{P}(t, t_0 \mid v_{i,j} = 1, E_2) \nonumber
    \\ = &\mathbb{P}(t, t_0 \mid v_{i,j} = 0, E_3) \nonumber
    \\ = &\mathbb{P}(t, t_0 \mid v_{i,j} = 1, E_3) \nonumber
    \\ = & \mbb{P}(E_{t, t_0}),
\end{align}

\begin{align}\label{align: combination 3}
    \mathbb{P}(t, t_0 \mid v_{i,j} = 0, E_5) \nonumber
     = &\mathbb{P}(t, t_0 \mid v_{i,j} = 1, E_1) \nonumber
    \\ = &\mathbb{P}(t, t_0 \mid v_{i,j} = 1, E_4) \nonumber
    \\ = & \mbb{P}(E_{t -1, t_0}).
\end{align}
Substituting \eqref{align: combination 1}--\eqref{align: combination 3} into \eqref{align: lr expression} leads to
\begin{align}\label{eqn: mij = P(E)}
m_{i,j} = &\big[\mbb{P}(E_{t-1, t_0-1})(p_1 + p_4) + \mbb{P}(E_{t, t_0})(p_2 + p_3) \nonumber
    \\&\quad  + \mbb{P}(E_{t-1, t_0})p_5 \big] \nonumber
    \\& \big/ \big[ \mbb{P}(E_{t-1, t_0})(p_1 + p_4) \nonumber
    \\&\quad  + \mbb{P}(E_{t, t_0})(p_2 + p_3) + \mbb{P}(E_{t-1, t_0-1}) p_5 \big].
\end{align}
By further substituting \eqref{eqn: P(Ett0)} into  \eqref{eqn: mij = P(E)}  and simplifying the result, the proof of \eqref{eqn: mij} (as well as Proposition \ref{proposition: mij}) is completed.

\bibliographystyle{IEEEtran}
\bibliography{myreference}

\end{document}